\documentclass[a4paper,onecolumn, 11pt, accepted=2024-04-11]{quantumarticle}
\pdfoutput=1
\usepackage{graphicx}
\usepackage{fullpage,latexsym,amsmath,mathtools,amsfonts,xcolor}
\usepackage{amssymb}
\usepackage{amsthm}
\usepackage{wrapfig}
\usepackage{subfigure}
\usepackage[hidelinks]{hyperref}
\usepackage[font=small,skip=0pt]{caption}
\usepackage{slashbox}
\usepackage{natbib}
\usepackage{doi}

\usepackage{hyperref}
\hypersetup{
    colorlinks,
    linkcolor={blue!100!black},
    citecolor={blue!100!black},
}

\def\01{\{0,1\}}

\newcommand{\ket}[1]{|#1\rangle}

\newcommand{\poly}{\mbox{\rm poly}}
\newcommand{\polylog}{\mbox{\rm polylog}}

\newtheorem{definition}{Definition}
\newtheorem{theorem}{Theorem}
\newtheorem{problem}{Problem}
\newtheorem{lemma}[theorem]{Lemma}
\newtheorem{proposition}[theorem]{Proposition}
\newtheorem{corollary}[theorem]{Corollary}

\newtheorem{fact}{Fact}

\renewcommand{\doitext}{DOI:}
\newcommand{\mydoi}[1]{\href{https://doi.org/#1}{\doitext\discretionary{}{}{}#1}}

\renewcommand{\doi}[1]{\mydoi{#1}}
\bibpunct[,]{[}{]}{,}{n}{,}{,}

\newcommand{\aff}{\mathsf{Aff}}

\newcommand{\be}{\begin{equation}}
\newcommand{\ee}{\end{equation}}
\newcommand{\ra}{\rangle}

\newcommand{\ba}{\begin{array}}
\newcommand{\ea}{\end{array}}

\newcommand{\calL}{{\cal L }}

\newcommand{\calC}{{\cal C }}
\newcommand{\calS}{{\cal S }}

\newcommand{\FF}{\mathbb{F}}

\newcommand{\ZZ}{\mathbb{Z}}

\begin{document}

\title{Generating $k$ EPR-pairs from an $n$-party resource state}

\author{Sergey Bravyi}
\affiliation{IBM Quantum, IBM T.J.\ Watson Research Center, Yorktown Heights, NY 10598, USA.} 
\email{sbravyi@us.ibm.com}

\author{Yash Sharma}
\affiliation{Rutgers University}
\email{yash.sharma@rutgers.edu}

\author{Mario Szegedy}
\affiliation{Rutgers University}
\email{szegedy@cs.rutgers.edu}

\author{Ronald de Wolf}
\affiliation{QuSoft, CWI and University of Amsterdam, the Netherlands.}
\thanks{Partially supported by the Dutch Research Council (NWO) through Gravitation-grant Quantum Software Consortium, 024.003.037.}
\email{rdewolf@cwi.nl}

\maketitle

\begin{abstract}
Motivated by quantum network applications over classical channels, we initiate the study of $n$-party resource states from which LOCC protocols can create EPR-pairs between any $k$ disjoint pairs of parties. 
We give constructions of such states where $k$ is not too far from the optimal $n/2$ while the individual parties need to hold only a constant number of qubits.
In the special case when each party holds only one qubit, we describe a family of $n$-qubit states with $k$ proportional to $\log n$ based on Reed-Muller codes,
as well as small numerically found examples for $k=2$ and $k=3$.
We also prove some lower bounds, for example showing that if $k=n/2$ then the parties must have at least $\Omega(\log\log n)$ qubits each.
\end{abstract}

\section{Introduction}

\subsection{Generating EPR-pairs from a resource state}

Quantum communication networks combine several quantum computers to enable them to solve interesting tasks from cryptography, communication complexity, distributed computing etc.
Building a large-scale quantum communication network is a daunting task that will take many years, but networks with a few small quantum computers are under construction and may start to appear in the next few years~\cite{WEH:qinternet}.

\setlength{\intextsep}{-11pt}%
\begin{wrapfigure}[9]{r}{0.2\textwidth}
\centering
\includegraphics[width=0.2\textwidth]{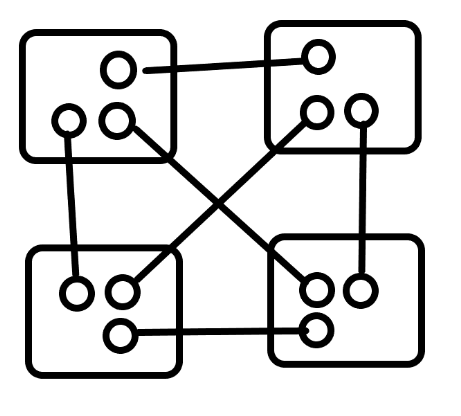}
  \caption{State of $n=4$ parties
  with 6 EPR-pairs} \label{fig:allpairs}
\end{wrapfigure}

These networks are either based on channels that physically communicate quantum states, or rely on classical communication in tandem with shared entanglement, or a combination of both. Communication over classical channels cannot increase entanglement, so in the absence of quantum channels we have to rely on prior entangled states. For example, if two parties share an EPR-pair, $\frac{1}{\sqrt{2}}(\ket{00}+\ket{11})$, then one party can transmit (``teleport'') a qubit to the other via two classical bits of communication, consuming the EPR-pair in the process~\cite{teleporting}. 
If we want to enable many qubits to be sent in this way, then we could start with an entangled state where each pair among the $n$ parties shares its own EPR-pair. This would allow any pair to exchange a qubit, but would require us to start with a rather large initial entangled state of $\binom{n}{2}$ EPR-pairs, and each of the $n$ parties would need to hold $n-1$ qubits (see Figure~\ref{fig:allpairs} for $n=4$).

Now suppose that we know in advance that only some $k$ pairs out of the $n$ parties will be required to exchange a qubit, but we do not know in advance what those $k$ pairs are.
In this paper we study what initial $n$-party resource states are sufficient or necessary to achieve this task.
Note that an EPR-pair between two parties allows them to exchange a qubit in either direction via local operations and classical communication (LOCC), and conversely the ability to exchange a qubit between two parties allows them to share an EPR-pair (one party locally creates the EPR-pair and sends one of its qubits to the other party). This shows that the ability to exchange qubits between any $k$ disjoint pairs of parties is essentially equivalent to the ability to establish EPR-pairs between any $k$ disjoint pairs.
We focus on the latter task in this paper.

We call an $n$-party state $\ket{\psi}$ \emph{$k$-pairable} if for every $k$ disjoint pairs $\{a_1,b_1\},\ldots,\{a_k,b_k\}$ of parties, there exists an LOCC protocol that starts with $\ket{\psi}$ and ends up with a state where each of those $k$ pairs of parties shares an EPR-pair.
%
The local quantum operations and the classical communication are free, but we do care about the number of qubits per party in $\ket{\psi}$: the fewer the better.
The same resource state $\ket{\psi}$ has to work for every possible $k$-pairing,
%
so it is fixed before the pairing task is given.
%
For example, the
$n$-qubit GHZ-state
\[
\frac{1}{\sqrt{2}}(\ket{0^n}+\ket{1^n})
\]
is 1-pairable: in order to obtain an EPR-pair between two parties Alice and Bob, the other $n-2$ parties can measure their qubit in the Hadamard basis and communicate the classical measurement outcomes to Alice and Bob, who convert their remaining 2-qubit state into an EPR-pair if one of them (say Alice) does a $Z$-gate conditioned on the parity of the $n-2$ bits they received.

\vspace{5mm}
\begin{center}
\begin{figure}[h]
\centerline{
\includegraphics[width=6.5cm]{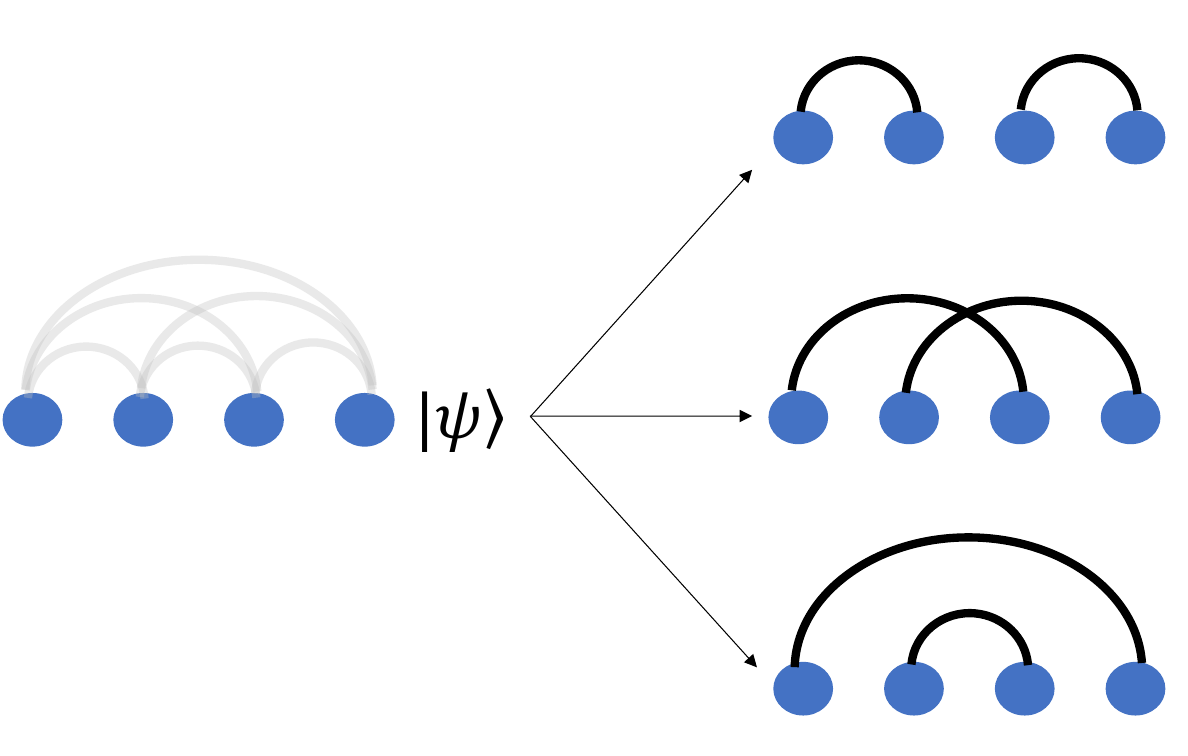}}
\caption{{A cartoon illustrating $k$-pairability.
An entangled resource state $\psi$ is initially distributed among $n$ parties.
     The parties can implement an arbitary LOCC protocol, that is,
     any local quantum operations and classical communication.
     The resource state $\ket{\psi}$ is called $k$-pairable if
     for any selection of $k$ disjoint pairs of parties, 
     there exists an LOCC protocol converting $\ket{\psi}$ to a collection of EPR-pairs shared by the selected pairs of parties. 
     In this example $k=2$ and $n=4$.
     The final EPR-pairs are indicated by solid arcs connecting the parties.
     Arrows indicate LOCC protocols converting $\ket{\psi}$ to the three possible desired final states, which correspond to the three possible ways of partitioning the $n=4$ vertices into $k=2$ disjoint pairs. Our goal is to maximize the pairability parameter $k$ while keeping the number of qubits per party as small as possible.}}  
     \label{fig:example}
 \end{figure}
\end{center}

The GHZ-example has the minimal possible 1 qubit per party, but unfortunately $k$ is only~1 there: we can only create 1 EPR-pair. We are interested in resource states that are $k$-pairable for larger $k\le n/2$. We give both upper and lower bounds for $k$-pairability, considering both the situation where we allow $m>1$ qubits per party (but not too many),
and the situation where we insist that each of the $n$ parties has only the minimal $m=1$ qubits.

For $k>1$ it is important to note that our definition of $k$-pairability requires the $k$ pairs to be disjoint, and does not allow overlapping pairs.
The main reason for this is that allowing $k$ overlapping pairs would in particular require us to be able to create a star graph of EPR-pairs, where one party shares EPR-pairs with $k$ other parties. This can only happen if that one party holds at least $k$ qubits.\footnote{To prove this, consider the bipartite state of party~1 vs the union of all other parties. If party~1 holds $m$ qubits, the initial Schmidt rank of the state will be $\leq 2^m$. The final state would have $k$ EPR-pairs between party~1 and the other parties, which requires Schmidt rank $\geq 2^k$. But LOCC cannot increase Schmidt rank, which implies $m\geq k$.}
This lower bound $m\geq k$ would rule out the constructions we have where $m$ is much smaller than $k$ (for example, $m=1$ vs $k=\Omega(\log n)$ in Section~\ref{sec:constrnoancilla}).
The $k$-pairability problem with $m\ll k$ is interesting from the practical standpoint since qubits, especially error-corrected logical qubits, built on top of multiple physical qubits, are expensive and we would prefer to have $n$-party $k$-pairable resource states with as few qubits per party as possible.

\subsection{Motivation}

Multi-party resource states are foundational for tasks such as distributed quantum computing \cite{raussendorf2003measurement, PhysRevLett.86.5188, VanMeter2014}, 
quantum secret sharing
\cite{PhysRevA.59.1829, VanMeter2014}, 
and multi-party quantum protocols \cite{CleveBuhrman1997, 2001_Fitzi}. One of the most utilized resources in quantum networking is EPR-pairs shared between two parties. Our model aims to capture scenarios involving multiple parties, where it is not known in advance between which parties the quantum resources should be deployed. This situation is similar to a call center, which must dynamically manage incoming calls. Our study concentrates on developing resource states that provide the flexibility to determine the call/connection structure \emph{after} the resource state has been created.

 Imagine for instance a scenario where multiple government agencies need to communicate securely and share sensitive information. The agencies are interconnected in a network where the communication partners may vary depending on the situation. 
 Similarly, a decentralized banking system might use quantum cryptography to secure transactions between multiple parties, ensuring that financial operations remain secure even as the network of participants changes, mirroring fluctuating market conditions.
 It is important to develop multi-party quantum resource states that can enable secure, flexible quantum communication systems over classical channels, allowing these agencies to establish secure communication channels dynamically, similar to switching call connections based on priority and need.

In this theoretical work we use several simplifying assumptions, most notably that EPR-pairs suffice as the primary resource required for applications.
The interesting follow-up work~\cite{CMP:smallkpairable,CCMPST:vertexminor}, discussed at the end of this section, subsequently removed this assumption.

\subsection{Our results 1: constructions of \texorpdfstring{$k$}{k}-pairable resource states}\label{ssec:introresults1}

In Section~\ref{sec:constrmultiple} we first study $k$-pairable resource states where each of the $n$ parties is allowed to have $O(1)$ qubits (hence $\ket{\psi}$ will have $O(n)$ qubits in total). We show that we can make $k$ as large as $n/\polylog(n)$ while each party holds only 10 qubits. Roughly, the idea is to take a special kind of $n$-vertex expander graphs that guarantee existence of $k$ edge-disjoint paths for any $k$ disjoint pairs, let each edge in the graph correspond to an EPR-pair, and  create the $k$ desired EPR-pairs via entanglement-swapping along the edge-disjoint paths. 
If we allow $m=O(\log n)$ qubits per party instead of $m=O(1)$, then we can construct $k$-pairable resource states with $k = n/2$, meaning that from our fixed resource state we can create EPR-pairs across any perfect matching of the $n$ parties into disjoint pairs.%
\footnote{If we allow one party to hold many more qubits than the others, then we could use a resource state~$\ket{\psi}$ corresponding to a star graph, where the central party shares an EPR-pair with each of the $n-1$ other parties, and uses entanglement-swapping (see the proof of Lemma~\ref{lemma:paths}) to link up the $k$ pairs as desired.
This $(2n-2)$-qubit state is $k$-pairable for the maximal $k=n/2$, and $n-1$ parties hold the minimal 1 qubit. However, the central party holds $n-1$ qubits and has to do all the work in obtaining the $k$-pairing. In the spirit of the small quantum networks of small quantum computers that we'll have in the near and medium-term future, we prefer constructions where none of the parties needs to hold many qubits.}
This result
essentially requires only 
classical off-the-shelf routing arguments.

Since qubits are expensive, especially when lots of error-correction is needed to protect them, we also look at what is possible when each party holds only 1 qubit, which is of course the bare minimum. In this case, we construct $n$-party (which in this case is the same as $n$-qubit) resource states for the case $k=1$ for arbitrary $n$ (this corresponds to the GHZ-state).
For $k\ge 2$ it is not clear that $k$-pairability is a property monotone in $n$.
What we have is that $k$-pairable states exist
for $k=2$ for $n=16$ and higher powers of~2 (we also give numerical evidence for the existence of a 2-pairable state on $n=10$ qubits); for $k=3$ for $n=32$ and higher powers of~2; and for arbitrary $k$ for $n=2^{3k}$ and higher powers of~2. 
These resource states will be superpositions over the codewords in a Reed-Muller code, and we use the stabilizer formalism to design LOCC protocols for obtaining the desired $k$ EPR-pairs from the resource state.
Our construction is efficient in the sense that all steps
in the LOCC protocol 
can be computed in time $\poly(n)$. To prove correctness of the protocol
we reduce the problem of EPR-pair generation to a version of
the polynomial regression problem: 
constructing a multi-variate $\FF_2$-valued 
polynomial of fixed degree that takes prescribed values at a given set of points.
One of our main technical contributions is developing tools for
solving a particular family of such polynomial regression problems.

Our protocols for generating EPR-pairs can be made fault-tolerant if each party encodes their $m$ qubits by some stabilizer-type quantum code~\cite{gottesman1998theory}.
Importantly, the encoded versions of our protocols only require logical Clifford gates and Pauli measurements.
These operations can be implemented transversally for many stabilizer codes, such as for instance the color codes of~\cite{bombin2006topological}.
Thus our protocols are well-suited for networks of small fault-tolerant
quantum computers and applications that require fast communication of qubits between $k$ arbitrary pairs of network nodes. 
A summary of our protocols can be found in Figure~\ref{fig:table}.

\vspace{1cm}
\begin{figure}[hbt]
\centerline{
\includegraphics[height=7.3cm]{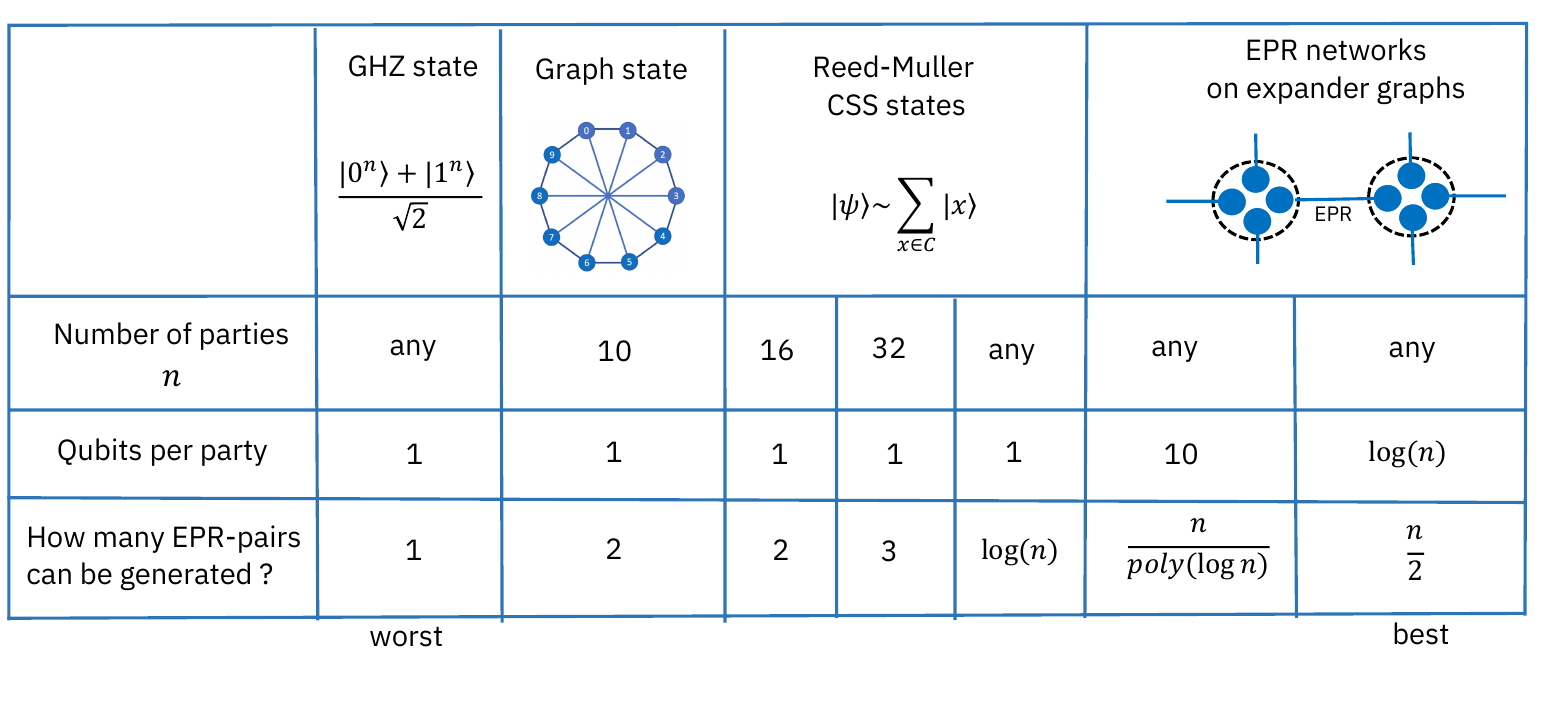}}
  \caption{Informal summary of our constructions. We consider resource states of different type 
  shared among $n$ parties such that each party holds a fixed number of qubits, ranging between $1$ and $\log{(n)}$.
  The last row shows the pairability parameter $k$---the number of EPR-pairs that can be 
 generated by LOCC starting from the respective resource state. For simplicity, we ignore constant factors in the $\log{(n)}$ scaling
 and ignore minor restrictions on the number of parties $n$ in certain cases, see Sections~\ref{sec:constrmultiple},\ref{sec:constrnoancilla} for details. Our proof of $k$-pairability is analytic in all cases except for $n=10$ and $32$ 
 where we provide only a computer-aided proof.} \label{fig:table}
\end{figure}
\vspace{1cm}

\subsection{Our results 2: obstructions}

Next we look at \emph{ob}structions, namely lower bounds on the achievable tradeoff between $n$, $k$, and~$m$. First consider the case where we can pair up any $k=n/2$ disjoint pairs. An ideal resource state would allow us to do this (i.e., be $n/2$-pairable) with only $m=1$ qubits per party. As mentioned above, we have shown that $k$-pairability with only 1 qubit per party is indeed achievable if $k\ll n/2$, but in Section~\ref{sec:obstructionskisn/2} we show it is not achievable if $k=n/2$: in that case $m=\Omega(\log\log n)$ qubits per party are needed. The proof is by an intricate dimension-counting argument, which gives upper and lower bounds on the dimension of the space of states that can be reached (with non-zero probability) by LOCC protocols on a fixed $nm$-qubit resource state~$\ket{\psi}$.
In Section~\ref{sec:obstructionsgeneralk} we extend this approach to the case of partial pairings, so where $k<n/2$, showing $m=\Omega\left(\log\left(\frac{k\log n}{n\log\log n}\right)\right)$ in this more general case. In particular, if 
$m=O(1)$ then $k$ can be at most 
$O\left(n\frac{\log\log n}{\log n}\right)=o(n)$, so achieving something close to complete pairability (i.e., $k=\Omega(n)$) requires a super-constant number of qubits per party.
Up to the power of the polylog, this matches our construction of $k$-pairable states with $k=n/\polylog(n)$ and $m=10$ qubits per party (Section~\ref{ssec:introresults1}).

We note here that our lower bounds apply to all possible LOCC-protocols, while our constructions are more lightweight, requiring only Pauli measurements and Clifford operations to create the $k$ EPR-pairs from the $n$-party resource state. We consider it a good thing that our upper bounds need only fairly restricted operations, while our lower bounds apply even to the general model with more powerful operations.

\subsection{Related and follow-up work}

To the best of our knowledge, the problem of what resource states allow LOCC protocols to construct EPR-pairs between any $k$ pairs of parties has not been studied before. However, we are aware of a number of related works, which we will briefly discuss here.\footnote{We won't go over the literature on quantum network coding (which assumes the ability to send qubits over specific edges) nor the massive experimental physics literature on preparation of entangled states on actual noisy hardware such as optical links and repeaters.}
These works can be organized into two categories.

\subsubsection*{Entanglement routing assisted by quantum communication.} Here some parties are allowed to exchange qubits in addition to performing LOCC
on the initial resource state.

Schoute et al.~\cite{schouteetal:shortcuts} consider quantum networks where parties can create EPR-pairs with their immediate neighbors and then use entanglement-swapping combined with efficient routing algorithms to create desired long-distance entanglement. This differs from our approach in allowing the ability to create new EPR-pairs when needed (which requires quantum communication), while we allow only LOCC starting from one fixed entangled resource state.

Hahn, Pappa, and Eisert~\cite{hahn2019quantum} also study a quite similar problem to ours, but starting from a network where some parties are linked via a quantum channel, while some other parties are not (directly) linked at all.
In addition to efficiently generating EPR-pairs they also study generating GHZ-states between specified parties.

Pant et al.~\cite{pant2019routing} study how a network whose nodes are connected via lossy optical links and have limited quantum processing capabilities,  can obtain EPR-pairs simultaneously between many pairs of nodes; their limitations per node are analogous to our goal of having only few qubits per party, but they allow quantum communication while we allow only classical communication.

\subsubsection*{Restricted variants of $k$-pairability.} 
Here the parties are only allowed to perform  LOCC on
the initial 
resource state. The parties may be able to generate $k$ EPR-pairs for some but not all choices of such pairs.

Miguel-Ramiro, Pirker, and D\"ur~\cite{MPD:networks}
consider resource states interpolating between the two extreme cases discussed in our  introduction: the GHZ state shared among $n$ parties and $\binom{n}{2}$
EPR states shared between each pair of parties. This work proposed clustering
and merging algorithms that produce resource states with the desired functionality. 
However, these methods do not appear to provide $k$-pairable resource states with few qubits per party.

Du, Shang, and Liu~\cite{DSL:multipoint} study a problem similar to ours but starting from resource states that consist only of pre-shared EPR-pairs between adjacent parties in a given network. Like us, they use entanglement-swapping to create EPR-pairs between distant parties.

Contreras-Tejada, Palazuelos, and de Vicente~\cite{CPV:noisynetworks} gave similar constructions as we gave in Section~\ref{sec:constrmultiple} (with EPR-pairs on the edges of an $n$-vertex graph), but focus primarily on the question for what type of graphs the long-range entanglement survives constant amounts of noise on the edges.

Illiano et al.~\cite{10008516} study 1-pairable states with the additional property that the identity of the one pair that ends up sharing an EPR-pair remains unknown to the other $n-2$ parties (in fact one can get this easily from the $n$-party GHZ-state if the other parties broadcast their measurement outcomes to everyone rather than sending it only to the two parties that want an EPR-pair).

Meignant, Markham, and Grosshans~\cite{MMG:distributing} and Fischer and Townsley~\cite{fischer&townsley21} studied what is roughly a partial ``dual'' of our problem: how many EPR-pairs between which parties of a given $n$-party network are necessary and sufficient to generate a classically given $n$-party graph state?

Dahlberg, Helsen, and Wehner~\cite{dahlberg2020transforming} show that it is NP-complete to decide whether a classically given $n$-party stabilizer state can be transformed into a set of EPR-pairs on specific qubits using only single-qubit Clifford operations, single-qubit Pauli measurements and classical communication (such protocols are more restricted than the LOCC we allow in our paper). They also give some algorithms to do the transformation in some special cases~\cite{Dahlberg_2020}.

\subsubsection*{Follow-up work.}
Our construction in Section~\ref{sec:constrnoancilla} has $m=1$ and $k=\Theta(\log n)$, so the number of qubits it uses to achieve $k$-pairability with one qubit per party, scales as $n=2^{\Omega(k)}$. Claudet, Mhalla, and Perdrix~\cite{CMP:smallkpairable}
recently improved this exponential scaling to a polynomial one: they show the existence of $k$-pairable  graph states with one qubit per party, using only $n = O(k^3 (\log k)^3)$ qubits (=parties). Their proof uses the  probabilistic method, so is not as constructive as ours. However, very recently this was improved by Cautr\`et, Claudet, Mhalla, Perdrix, Savin, and Thomass\'e~\cite{CCMPST:vertexminor} to a probabilistic proof with $n = O(k^2)$ qubits and an explicit construction with $n = O(k^4)$ qubits.

\section{Constructions with multiple qubits per party}\label{sec:constrmultiple}

In this section we combine classical network-routing strategies
and the standard entanglement-swapping protocol to 
construct $n$-party $k$-pairable resource states 
with $k$ nearly linear in $n$, 
such that each party holds at most $m=O(1)$ qubits.
Increasing the number of qubits per party from a constant to $m=O(\log{n})$ yields maximally pairable
resource states with $k=n/2$.

 Suppose $G=(V,E)$ is a graph with $n$ vertices
$V=\{1,2,\ldots,n\}$. Vertex $i\in V$ represents the $i$-th party. 
We place two qubits at every edge $(i,j)\in E$ such that in total there are $n=2|E|$
qubits. Define an $n$-party resource state 
\[
\ket{\psi_G} = \bigotimes_{(i,j)\in E} \ket{\Phi^+_{i,j}},
\]
where $\ket{\Phi^+_{i,j}}$ is an EPR-pair located on edge $(i,j)$. 
The state $|\psi_G\rangle$ is shared among $n$ parties such that the two qubits
located on an edge $(i,j)\in E$ are assigned to the parties $i$ and $j$
who share the EPR-pair $|\Phi^+_{i,j}\rangle$. Thus each party shares one EPR-pair with each of its neighbors.
 Accordingly, each party holds at most $d$ qubits, where $d$ is the maximum vertex degree of $G$.

\begin{lemma}
\label{lemma:paths}
The resource state $|\psi_G\rangle$ is $k$-pairable if for any choice of $k$ disjoint pairs of vertices
$\{a_1,b_1\},\ldots,\{a_k,b_k\}$ in the graph $G$, there exist $k$ edge-disjoint paths $P_1,\ldots,P_k\subseteq E$ such that the path $P_i$ connects
vertices $\{a_i,b_i\}$.
\end{lemma}

\begin{proof}
Suppose Charlie shares an EPR-pair with Alice and another EPR-pair with Bob. The following well-known entanglement-swapping protocol uses LOCC to create an EPR-pair between Alice and Bob. First, Charlie measures the parity of his two qubits in the standard basis $\{\ket{0},\ket{1}\}$, sends the 1-bit measurement outcome to Bob, and conditioned on it he applies a $\sigma^x$ (bitflip) on his second qubit and Bob applies a $\sigma^x$ to his qubit. This results in a 4-qubit GHZ-state $\frac{1}{\sqrt{2}}(\ket{0000}+\ket{1111})$. Now Charlie measures each of his two qubits in the Hadamard basis $\{\ket{+},\ket{-}\}$, sends the parity of the two outcomes to Bob, who conditioned on that bit applies a $\sigma^z$ (phaseflip) to his qubit. It may be verified that now Alice and Bob share an EPR-pair.

The creation of the $k$ EPR-pairs using the $k$ edge-disjoint paths is now fairly straightforward: the parties on the path from $a_i$ to $b_i$ use the EPR-pairs with their neighbors on the path to create an EPR-pair between $a_i$ and $b_i$ via entanglement-swapping. Because the $k$ paths are edge-disjoint, no edge (=EPR-pair) is used more than once.
\end{proof}

Below it will be convenient to relax the edge-disjointness condition  in Lemma~\ref{lemma:paths} 
and consider pairability by nearly edge-disjoint paths. 
More precisely, suppose $p\ge 1$ is an integer. Consider a resource state
$|\psi_G\rangle^{\otimes p}$ such that each copy of $|\psi_G\rangle$ is shared among $n$ parties as specified above. Then each party holds at most $pd$ qubits, where $d$
is  the maximum vertex degree of $G$. 
Each party shares $p$ EPR-pairs with each of its neighbors.
 An immediate corollary of Lemma~\ref{lemma:paths} is the following.
 
\begin{corollary}
\label{corol:paths}
The resource state $|\psi_G\rangle^{\otimes p}$ is $k$-pairable
if for any choice of $k$ disjoint pairs of vertices
$\{a_1,b_1\},\ldots,\{a_k,b_k\}$ in the graph $G$, there exist $k$
paths $P_1,\ldots,P_k\subseteq E$ such that the path $P_i$ connects
vertices $\{a_i,b_i\}$ and each edge of $G$ belongs to at most $p$ paths.
\end{corollary}

To keep the number of qubits per party small, we would like the graph $G$ to have 
a small vertex degree and, at the same time, allow vertex pairability by (nearly) edge-disjoint paths
for any choice of $k$ disjoint vertex pairs. We would like to maximize the pairability parameter $k$
while keeping the vertex degree $d$ as small as possible.
Luckily, the problem of constructing such graphs has been already studied 
due to its importance for classical communication networks. 
A graph $G=(V,E)$ is said to have \emph{edge expansion} $h$
if for any subset of vertices $S\subseteq V$ with $|S|\le |V|/2$, the number of edges that
have exactly one endpoint in $S$ is at least $h |S|$.
We shall use the following fact.

\begin{fact}[\bf Broder, Frieze, Upfal~\cite{broder1994existence}]
\label{fact:paths}
For any constants $d\ge 3$ and $h>1$ there exists a constant $c>0$
such that the following is true.
Suppose $G$ is an  $n$-vertex $d$-regular graph  with edge expansion at least $h$.
Then for any choice of $k\le n/\log^c{(n)}$ disjoint vertex pairs in $G$ there exists
a family of paths $P_1,\ldots,P_k$ connecting the chosen pairs of vertices such that every
edge of $G$ belongs to at most two paths. 
These paths can be found in time $\poly(n)$. 
\end{fact}

It is known~\cite{bollobas1988isoperimetric} that $d$-regular graphs with edge expansion $h>1$ exist
for any constant $d\ge 5$ and all large enough $n$. Thus Corollary~\ref{corol:paths} and Fact~\ref{fact:paths} imply that for 
all large enough $n$ there exist $k$-pairable resource states with $k=n/\polylog n$ and at most 10 qubits per party. 

Let us say that  a graph $G$ is path-pairable 
if the number of vertices $n$ is even and 
the condition of Lemma~\ref{lemma:paths} holds 
for $k=n/2$. We shall need the following fact
stated as Corollary~2 in~\cite{gyHori2017note}.
\begin{fact}[\bf Gy{\H{o}}ri, Mezei,  M{\'e}sz{\'a}ros~\cite{gyHori2017note}]
\label{fact:pairable}
For any integer $q\ge 1$
there exists a $d$-regular $n$-vertex path-pairable graph with $d=18q$ and $n=18^q$. 
\end{fact}
Combining Lemma~\ref{lemma:paths} and Fact~\ref{fact:pairable} we infer that $m=O(\log n)$ qubits per party suffices for complete pairings, in contrast with the naive resource state where every one of the $\binom{n}{2}$ pairs shares an EPR-pair and hence each party holds $m=n-1$ qubits.

\begin{corollary}\label{cor:mis10}
There exists
a family of $n$-party $(n/2)$-pairable resource states 
with $m=18\log_{18}(n)\approx 4.3 \log_2(n)$ qubits per party. 
\end{corollary}

\section{Constructions that use only one qubit per party}\label{sec:constrnoancilla}


In this section we study $k$-pairability  of  $n$-party quantum states under the most stringent restriction: each party holds only one qubit
(obviously, $k$-pairability with $k\ge 1$ is impossible if some party has no qubits).

We have already seen that the $n$-qubit GHZ-state shared by $n$ parties is $1$-pairable.
Naively, one might think that the GHZ-example is already best-possible and $k=1$ is as far as one can get
with one qubit per party. Surprisingly, this naive intuition turns out to be wrong.
Here we 
give examples of $k$-pairable states with one qubit per party for an arbitrary $k$.
We choose the resource state $\ket{\psi}$
as the uniform superposition of codewords of a suitable linear code $\calC$ of codelength~$n$.
The GHZ-example $|\psi\ra = \frac{1}{\sqrt{2}}(|0^n\ra+|1^n\ra)$
with $k=1$ is the special case with the repetition code $\calC=\{0^n,1^n\}$.

To achieve $k$-pairability for $k\ge 2$ we choose $\calC$ as the Reed-Muller code $\mathsf{RM}(k-1,m)$ with 
a suitable parameter $m$, see below for 
details.\footnote{We follow the standard notation for the parameter
$m$ of Reed-Muller codes. It should not be confused with the number of qubits per party, which equals~1 throughout this section.} The LOCC protocol converting $\ket{\psi}$ to the desired EPR-pairs
can be described by a
pair of disjoint subsets $X,Z\subseteq [n]$
such that all qubits contained in $Z$ and $X$ are measured in the standard basis $\{|0\ra,|1\ra\}$
and the Hadamard basis $\{|+\ra,|-\ra\}$ respectively.
The protocol creates EPR-pairs on $2k$ qubits contained in the complement $[n]\setminus (X\cup Z)$.
Here $|\pm \ra=(|0\ra\pm |1\ra)/\sqrt{2}$.
Finally, a Pauli correction
$\sigma^x$ or $\sigma^z$ is applied to each EPR qubit $a_1,\ldots,a_k$. The correction
depends on the measurement outcomes and requires classical communication from 
parties in $X\cup Z$ to parties $a_1,\ldots,a_k$.

Our construction is efficient in the sense that the
subsets of qubits $X$ and $Z$ can be computed in time $O(n)$ for any given choice of EPR qubits.
Furthermore, the initial resource state $\ket{\psi}$ can be prepared by a quantum circuit of size $O(n^2)$.
While describing the subsets $X$ and $Z$ is relatively simple, proving that the resulting LOCC
protocol indeed generates the desired EPR-pairs is considerably more complicated in the case $k\ge 2$,
as compared with the GHZ-example for $k=1$.
For resource states based on Reed-Muller codes 
$\mathsf{RM}(k-1,m)$, we will see below that 
the proof can be reduced to solving a polynomial regression problem: constructing a polynomial $f\, : \, \FF_2^m\to \FF_2$ of degree $k-1$ whose values $f(x)$ are fixed at a certain subset of points $x$.
The number of qubits 
$n = n(k)$ used by our construction is given by $n(2) = 16$, $n(3) = 32$, and $n(k) = 2^{3k}$ for $k\ge 4$
(note that the number of qubits is the same as the number of parties throughout this section).
While this scaling $n(k)$  may be far from optimal, the main value of our result
is demonstrating that $k$-pairability with 
an arbitrary $k$ is possible in principle even in the most restrictive setting
with one qubit per party.
To the best of our knowledge, this was not known prior to our work.
We leave as an open question whether $k$-pairable states based on Reed-Muller codes
can achieve a more favorable scaling $n(k)=\poly(k)$ or even  the scaling
$n(k)=O(k\,\polylog(k))$ that can be achieved if we allow 10 qubits per party instead of 1 (end of Section~\ref{sec:constrmultiple}). Such an improvement may require consideration of more general LOCC protocols
that use all three types of Pauli measurements, in the $\sigma^x,\sigma^y,\sigma^z$ bases.

Finally, we describe a numerically-found example of a  $10$-qubit $2$-pairable state with one qubit per party; this is more efficient than the 16-qubit $2$-pairable state from the above results.  This example is based on a stabilizer-type resource state and an LOCC protocol with
Pauli measurements. We also show that no stabilizer state with $n\le 9$ qubits is $2$-pairable using only Pauli measurements. In that sense our $10$-qubit example is optimal.

The rest of this section is organized as follows.
We introduce CSS-type resource states
and give sufficient conditions for $k$-pairability
of such states in Section~\ref{subs:css}.
Reed-Muller codes and their basic properties
are described in 
Section~\ref{subs:RM}. 
We define resource states based on Reed-Muller codes
and describe our LOCC protocol 
for generating EPR-pairs
in Section~\ref{subs:XZ}.
A proof of $k$-pairability
for $k=2,3$, and for an arbitrary $k$ is given
in Sections~\ref{subs:2-pairable}, \ref{subs:3-pairable},
and~\ref{subs:k-pairable} respectively.
Finally, we describe the $10$-qubit $2$-pairable
example in Section~\ref{subs:10qubit}.

\subsection{Pairability of CSS stabilizer states}
\label{subs:css}

To describe our construction we need more notation.
Let $\FF_2^n=\{0,1\}^n$ be the $n$-dimensional vector space over~$\FF_2$.
Given a vector $f\in \FF_2^n$ and a bit index $j$,  let  $f(j)\in \{0,1\}$ be the $j$-th bit of $f$.
We write $f\cdot g = \sum_{j=1}^n f(j) g(j)$ for the dot product of vectors $f,g\in \FF_2^n$.
Unless stated otherwise, addition of binary vectors and the dot product  are computed  modulo two. 
The weight of a vector $f\in \FF_2^n$ is the number of bits $j$ such that $f(j)=1$.
A  linear code of length $n$ is simply a linear subspace $\calC\subseteq \FF_2^n$.
Vectors $f\in \calC$ are called codewords. The code is said to have distance $d$
if any nonzero codeword has weight at least $d$.
The dual code of $\calC$, denoted $\calC^\perp$, is the subspace 
of vectors $f\in \FF_2^n$
such that $f\cdot g =0$ for all $g\in \calC$.
An affine subspace of dimension $d$ is a set of vectors $\{f+h\, \mid \, f\in \calC\}$,
where $\calC\subseteq \FF_2^n$ is a linear subspace of dimension $d$ and $h\in \FF_2^n$ is some fixed vector.

Suppose our $n$-qubit resource state $\ket{\psi}$ has the form 
\be
|\psi\ra =|\calC\ra :=\frac1{\sqrt{|\calC|}}  \sum_{f\in \calC} |f\ra,
\ee
where $\calC\subseteq \FF_2^n$ is a linear code.
Such states are known as Calderbank-Shor-Steane (CSS) stabilizer states~\cite{calderbank1996good,steane1996multiple,calderbank1998quantum}.
It is well-known that  the state $|\calC\ra$ can be prepared by a quantum circuit of size $O(n^2)$
for any linear code $\calC$, 
see for instance~\cite{aaronson2004improved}.
We begin by deriving a sufficient condition under which a CSS stabilizer state
is $k$-pairable. Below  we assume that each of the $n$ parties holds only one qubit.

\begin{lemma}[\bf Pairability of CSS stabilizer states]
\label{lemma:css}
Suppose $\calC\subseteq \FF_2^n$ is a linear code.
Suppose for any set of $k$ disjoint pairs of qubits
$\{a_1,b_1\},\ldots,\{a_k,b_k\}$
 there exists a partition of the $n$ qubits into three disjoint subsets
\be
\{1,2,\ldots,n\} = EXZ
\ee
such that  $E=\{a_1,b_1,\ldots,a_k,b_k\}$ and
the following conditions hold for all $i=1,2,\ldots,k$:
\begin{description}
\item[\bf CSS1:] $\exists \; f\in \calC$ such that $f(a_i)=f(b_i)=1$ and for all $p\in EZ\setminus \{a_i,b_i\}$: $f(p)=0$
\item[\bf CSS2:] $\exists \; \bar{f}\in \calC^\perp$ such that $\bar{f}(a_i)=\bar{f}(b_i)=1$ and for all $p\in EX\setminus \{a_i,b_i\}$: $\bar{f}(p)=0$
\end{description}
Then the $n$-qubit state $|\calC\ra=\frac1{\sqrt{|\calC|}} \sum_{f\in \calC} |f\ra$ is $k$-pairable. 
\end{lemma}

Here and below we use shorthand set union notation $XY\equiv X\cup Y$ whenever $X$ and $Y$ are disjoint sets.
The desired EPR-pairs can be generated in three steps.
First, each qubit $p\in Z$ is measured in the standard basis $\{|0\ra,|1\ra\}$ and each qubit $p\in X$ is measured
in the Hadamard basis $\{+\ra,|-\ra\}$. Next, each party $p\in XZ$ broadcasts their binary measurement outcome to $a_1,\ldots,a_k$.
Finally, a Pauli correction is applied to each qubit $a_i$; this may depend on the measurement outcomes.

\begin{proof}[\bf Proof of Lemma~\ref{lemma:css}]
We assume some familiarity with the stabilizer formalism~\cite{calderbank1996good,gottesman1998heisenberg,nielsen2002quantum}.
Let $\sigma^x_j$ and $\sigma^z_j$ be single-qubit Pauli operators acting on the $j$-th qubit
tensored with the identity on all other qubits.
The resource state $|\calC\ra$ has Pauli 
stabilizers\footnote{Note that $\sigma^x(f)|\calC\ra=\frac1{\sqrt{|\calC|}}\sum_{g\in \calC}\sigma^x(f)|g\ra =
\frac1{\sqrt{|\calC|}}\sum_{g\in \calC}|f+g\ra =
\frac1{\sqrt{|\calC|}}\sum_{g'\in \calC}|g'\ra=
|\calC\ra$ 
for any $f\in \calC$.
Here the third equality uses the fact that $\calC$ is a linear code, i.e., a subspace.
Furthermore, $\sigma^z(\bar{f})|\calC\ra=|\calC\ra$ since $\sigma^z(\bar{f})|g\ra=(-1)^{\bar{f}\cdot g}|g\ra=|g\ra$ for any $\bar{f}\in \calC^\perp$ and
$g\in \calC$.} 
\[
\sigma^x(f)\equiv \prod_{p\, : \, f(p)=1} \sigma^x_p \quad \mbox{for} \quad f\in \calC
\]
and
\[
\sigma^z(\bar{f})\equiv \prod_{p\, : \, \bar{f}(p)=1} \sigma^z_p \quad \mbox{for} \quad \bar{f}\in \calC^\perp.
\]
Thus we have $\sigma^x(f)|\calC\ra=\sigma^z(\bar{f})|\calC\ra=|\calC\ra$.
Suppose $f$ and $\bar{f}$ obey conditions CSS1, CSS2 for some pair $\{a_i,b_i\}$.
Let $m_p=\pm 1$ be the measurement outcome on a qubit $p\in XZ$.
 Condition
CSS1 implies that the stabilizer $\sigma^x(f)$ commutes with Pauli operators $\sigma^z_p$ on qubits $p\in Z$, which are measured in the standard basis. Thus $\sigma^x(f)$ 
and $\{ m_p \sigma^x_p \, \mid p\in X\}$ are
stabilizers of the final state after the measurement.
We infer that the final state is stabilized by 
\[
S^x_i \; = \; \sigma^x(f)  
\prod_{p\in X\, : \, f(p)=1} m_p \sigma^x_p \;\; = \;\; 
\sigma^x_{a_i} \sigma^x_{b_i} 
\prod_{p\in X\, : \, f(p)=1} \; m_p.
\]
Here the second equality follows from CSS1.
Likewise, CSS2 implies that  the stabilizer $\sigma^z(\bar{f})$ commutes with Pauli operators $\sigma^x_p$ on qubits $p\in X$, which are measured in the Hadamard basis. Thus $\sigma^z(\bar{f})$
and $\{ m_p \sigma^z_p \, \mid p\in Z\}$
are  stabilizers of the final state.
We infer that the final state is stabilized by 
\[
S^z_i \; = \; \sigma^z(\bar{f})  \prod_{p\in Z\, : \, \bar{f}(p)=1} m_p \sigma^z_p 
\;\; = \;\; \sigma^z_{a_i} \sigma^z_{b_i} 
\prod_{p\in Z\, : \, \bar{f}(p)=1} \; m_p.
\]
Here the second equality follows from CSS2.
This is only possible if the final state
contains an EPR-pair on the qubits $\{a_i,b_i\}$, up to 
a Pauli correction $\sigma^x_{a_i}$  and/or $\sigma^z_{a_i}$. The correction can be applied via LOCC if each party $p\in XZ$
broadcasts their  measurement outcome $m_p$ to all parties $a_1,\ldots,a_k$.
\end{proof}
Thus it suffices  to show that for any $k\ge 1$ one can choose a sufficiently large $n$ and 
a linear code $\calC\subseteq \FF_2^n$ that satisfies $k$-pairability conditions CSS1 and CSS2
of Lemma~\ref{lemma:css}.
Below we will choose $\calC$ from the family of Reed-Muller codes~\cite{macwilliams1977theory} to achieve this.

\subsection{Reed-Muller codes}
\label{subs:RM}

First, let us record the definition and some basic properties of Reed-Muller codes.
Let $m\ge 1$ be an integer. 
A Boolean function $f\, : \, \FF_2^m\to \FF_2$ can be considered as a binary vector of length $n=2^m$
which lists the function values $f(x)$ for all inputs $x\in \FF_2^m$ in some fixed  (say, the lexicographic) order.
 For example, if $m=2$ and $f(x)=1+x_1x_2$ then 
we can consider $f$ as a length-$4$ binary vector 
\[
[f(00),\, f(10),\, f(01),\, f(11)] = [1,1,1,0].
\]
Reed-Muller code $\mathsf{RM}(r,m)$ has length $n=2^m$ and 
its codewords are the $n$-bit vectors associated with  $m$-variate degree-$r$ polynomials
$f\, : \, \FF_2^m \to \FF_2$. 
One can choose generators of $\mathsf{RM}({r},{m})$
as a set of monomials 
$\prod_{j\in S} x_j$ where $S$ runs over all subsets of $[m]$ of size
at most $r$. The monomial  associated with the empty set $S=\emptyset$
is the constant-1 function. 
For example, $\mathsf{RM}(0,m)$ is the repetition code of length $n=2^m$
since there are only two degree-$0$ polynomials: $f(x)\equiv 1$ and $f(x)\equiv 0$.
We shall use the following  facts. 
\begin{fact}[\bf Code parameters]
\label{fact:params}
The  Reed-Muller code $\mathsf{RM}(r,m)$ has dimension
\[
D(r,m) = \sum_{p=0}^r \binom{m}{p}
\]
and distance $2^{m-r}$.
\end{fact} 
\begin{fact}[\bf Dual code]
\label{fact:duality}
Suppose $0\le r<m$. Then $\mathsf{RM}({r},{m})^\perp = \mathsf{RM}({m-r-1},{m})$.
\end{fact}
\begin{fact}[\bf Affine invariance]
\label{fact:symmetry}
Suppose  $A\in \FF_2^{m\times m}$
is an invertible matrix, and $b\in \FF_2^m$ is a vector. 
If $f\, : \, \FF_2^m\to \FF_2$ is a degree-$r$ polynomial, 
then the function
$f'(x)=f(Ax+b)$ is also a degree-$r$ polynomial. The map $f\to f'$ is
a bijection of the set of all $m$-variate degree-$r$ polynomials. 
\end{fact}
For the proof of  Facts~\ref{fact:params}, \ref{fact:duality}, \ref{fact:symmetry}, see e.g.\ Chapter 13 of~\cite{macwilliams1977theory}.
As a consequence of Fact~\ref{fact:symmetry},
the resource state $|\calC\ra$ with 
$\calC=\mathsf{RM}(r,m)$ is invariant under a permutation of the $n=2^m$ qubits defined as $W_\varphi|f\ra=|f'\ra$, where $f'(x)=f(\varphi(x))$ for all $x\in \FF_2^m$.
Here $\varphi\, : \, \FF_2^m\to \FF_2^m$
is any invertible affine map. In other words, $W_\varphi|\calC\ra=|\calC\ra$. This generalizes the symmetry of the $n$-qubit GHZ-state which is invariant under any permutation of the $n$ qubits. 

Recall that minimum-weight codewords of a linear code are non-zero codewords  whose
weight equals the code distance.
\begin{fact}[\bf Codewords from affine subspaces]
\label{fact:affine}
A vector $f\in \FF_2^n$ is a minimum-weight codeword of $\mathrm{RM}(r,m)$ if and only if
the support of $f$ is an $(m-r)$-dimensional affine subspace of $\FF_2^m$.
\end{fact}
For the proof see, e.g., Proposition~2 and Corollary~4 in~\cite{assmus1992reed}.
We shall see that verification of conditions CSS1 and CSS2 of Lemma~\ref{lemma:css} 
with $\calC=\mathsf{RM}(r,m)$ can be reduced 
to (multiple instances of) the following problem. 
\begin{problem}[\bf Polynomial regression]
Find a degree-$r$ polynomial $f\, : \, \FF_2^m\to \FF_2$  that satisfies a system of equations
\be
f(x^i)=g_i \quad  \mbox{for} \quad i=1,\ldots,s
\ee
where $x^1,\ldots,x^s\in \FF_2^m$ are distinct points and $g_1,\ldots,g_s\in \{0,1\}$.
\end{problem}
 \begin{lemma}
\label{lemma:reg}
The polynomial regression problem has a solution $f$ if  at least one of the following conditions is satisfied:
(1) $s<2^{r+1}$, or (2) $s=2^{r+1}$ and $\sum_{i=1}^s g_i =0$.
\end{lemma}
Here the sum $\sum_{i=1}^s g_i$ is evaluated modulo two.
\begin{proof}
$\mathsf{RM}(r,m)^\perp = \mathsf{RM}(m-r-1,m)$ by Fact~\ref{fact:duality}.
The code $\mathsf{RM}(m-r-1,m)$ 
has distance $2^{r+1}$, see Fact~\ref{fact:params},
and thus every $2^{r+1}-1$ columns of its parity check matrix are linearly independent.
The parity check matrix $M$ of $\mathsf{RM}(m-r-1,m)$
is the generator matrix of its dual, $\mathsf{RM}(r,m)$, so the above 
implies that if $s<2^{r+1}$, then the rank of the matrix $M_{X}$ formed by the columns of $M$
with indices from $X = \{x^1,\ldots,x^s\}$
is $s$, so $\FF_{2}^X$ is in the span of the rows of~$M_{X}$.

If $s=2^{r+1}$, then the rank of $M_{X}$
is either $2^{r+1}$ or $2^{r+1}-1$. If $2^{r+1}$, we proceed as above.
If $2^{r+1}-1$, the only linear combination of the columns of $M_X$ 
that gives the zero-vector, is the sum of all columns of $M_X$,
and a vector $(g_1,\ldots,g_s)\in \FF_{2}^X$ can be generated 
from the rows of $M_X$ if and only if
$\sum_{i=1}^s g_i =0 \mod 2$.
\end{proof}

\subsection{Resource state and LOCC protocol}
\label{subs:XZ}

Our candidate $k$-pairable state is a CSS stabilizer state $|\calC\ra$ with 
\be
\calC=\mathsf{RM}(k-1,m)
\ee
and a suitable parameter $m=m(k)$.
To describe the subsets of qubits $X,Z\subseteq \FF_2^m$
satisfying conditions CSS1 and CSS2 of Lemma~\ref{lemma:css}, we need one extra piece of notation.

\begin{definition}
Suppose  $S\subseteq \FF_2^m$ is a non-empty subset. An affine subspace spanned by $S$,
denoted $\aff{(S)}$, is defined as 
\be
\aff{(S)} = \left\{ \sum_{v\in T} v \; \mid \; 
T\subseteq S \quad \mbox{and} \quad 
|T|=1{\pmod 2}\right\}.
\ee
\end{definition}

Thus $\aff{(S)}$ contains all vectors 
that can be written as a sum of an odd number of vectors
from~$S$.
For example, $\aff(\{a\})=\{a\}$, $\aff(\{a,b\})=\{a,b\}$, and 
$\aff(\{a,b,c\}) = \{a,\; b,\; c,\; a+b+c\}$.
Note that $|\aff(S)|=2^d$, where $d\le |S|-1$ is the dimension of $\aff{(S)}$.

\begin{figure}[ht]
\vspace{10mm}
\centerline{\includegraphics[height=5cm]{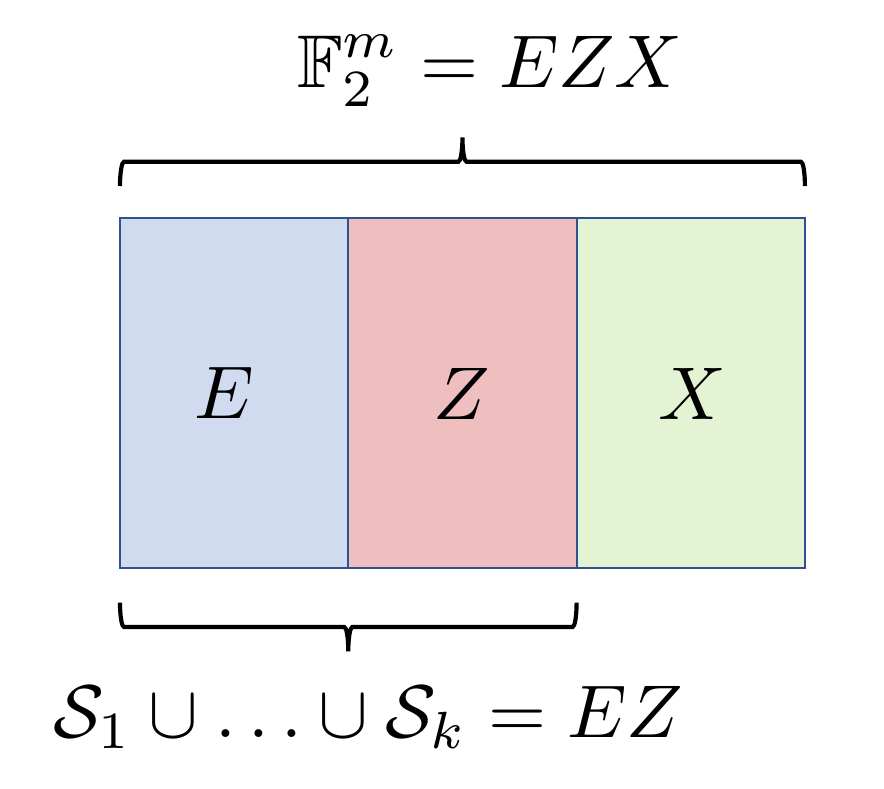}}
\caption{Measurement pattern for the resource
state $|\calC\ra$, where $\calC=\mathsf{RM}(k-1,m)$
and $\calC^\perp=\mathsf{RM}(m-k,m)$.
We consider $n=2^m$ qubits. Each qubit is labeled 
by an $m$-bit string.
Qubits are partitioned into three disjoint subsets, $EXZ$, where $E=\{a_1,b_1,\ldots,a_k,b_k\}$
is the set of EPR qubits, $Z$ is the set of qubits measured in the standard basis $\{|0\ra,|1\ra\}$ and
$X$ is the set of qubits measured in the Hadamard basis $\{|+\ra,|-\ra\}$. We choose $Z=(\calS_1\cup \ldots\cup \calS_k)\setminus E$, where the 
$\calS_i$ are $k$-dimensional affine subspaces 
of $\FF_2^m$, see Eq.~(\ref{Sk}). These subspaces are chosen such that $\calS_i\cap E = \{a_i,b_i\}$ for all $i$. We choose $X$ as the complement of $EZ$.
A codeword $\bar{f}\in\calC^\perp$ satisfying condition CSS2 for some pair of EPR
qubits $\{a_i,b_i\}$ is chosen as the characteristic function of the subspace $\calS_i$, that is,
$\bar{f}(x)=1$ if $x\in \calS_i$ and $\bar{f}(x)=0$
if $x\in \FF_2^m\setminus \calS_i$.
A codeword $f\in \calC$ satisfying condition CSS1
is constructed using the polynomial regression,
see Lemma~\ref{lemma:reg}.
\vspace{10mm}
}
\label{fig:EXZ}
\end{figure}

Let $n=2^m$ be the number of qubits.
Suppose our goal is to generate $k$ EPR-pairs on pairs of qubits
$\{a_1,b_1\},\ldots,\{a_k,b_k\}$. Define a subset of ``EPR qubits''
\[
E=\{a_1,b_1,\ldots,a_k,b_k\}
\]
and a family of $k$ affine subspaces
$\calS_1,\ldots,\calS_k\subseteq \FF_2^m$ such that 
\be
\label{Sk}
\calS_i=\aff{(\{a_i,b_i,c_1,c_2,\ldots,c_k\}\setminus \{c_i\})} \subseteq \FF_2^m
\ee
where $c_1,\ldots,c_k\in \FF_2^m$ 
are vectors that will be appropriately defined in
Sections~\ref{subs:2-pairable}, \ref{subs:3-pairable}, \ref{subs:k-pairable}.
The $c$-vectors may depend on the $a$'s and $b$'s. 
The set of EPR qubits $E$ is obviously
contained in the union of $\calS_1,\ldots,\calS_k$.
We choose the subsets of qubits $X$ and $Z$ in Lemma~\ref{lemma:css} as
\be
\label{XZ}
X = \FF_2^m \setminus  (\calS_1\cup \ldots \cup \calS_k) \quad  \mbox{and} \quad
Z = (\calS_1\cup \ldots \cup \calS_k)\setminus E
\ee
The subsets $E,X,Z$ are pairwise disjoint and 
$\FF_2^m =EXZ$.
We illustrate the relationships between  these sets  in Figure~\ref{fig:EXZ}.
In the GHZ-example one has $k=1$ and $\calS_1=\{a_1,b_1\}$. In this case $Z=\emptyset$
and $X=\FF_2^m\setminus \{a_1,b_1\}$, that is, the LOCC protocol requires only measurements
in the Hadamard basis.
In the rest of this section we prove 
that the vectors $c_1,\ldots,c_k$ in Eq.~(\ref{Sk}) can always be chosen
such that the subsets $X$ and $Z$ satisfy conditions CSS1 and CSS2 of Lemma~\ref{lemma:css}.

\subsection{\texorpdfstring{$2$}{2}-pairability}
\label{subs:2-pairable}

We now need to show how to choose $c_1,\ldots,c_k$.
We begin with the simple case $k=2$.

\begin{lemma}
Suppose $k=2$ and $m\ge 4$. Choose any vector $c\in \FF_2^m\setminus \aff{(E)}$
and let $c_1=c_2=c$. Then the subsets $X$, $Z$ defined in Eqs.~(\ref{Sk},\ref{XZ}) 
satisfy conditions CSS1 and CSS2 with $\calC=\mathsf{RM}(1,m)$.
\end{lemma}

\begin{proof}
Note that a vector $c$ as above exists since $|\aff{(E)}|\le 2^{|E|-1}=8$
and $|\FF_2^m|\ge 16$ for $m\ge 4$.
Specializing Eq.~(\ref{Sk}) to the case $k=2$ and $c_1=c_2=c$ one gets 
\be
\label{Sk2}
\calS_1 = \{a_1,b_1,c,a_1+b_1+c\} \quad\mbox{and} 
\quad \calS_2 = \{a_2,b_2,c,a_2+b_2+c\}.
\ee 
The assumption that $c\notin E$
implies that the $\calS_i$ are 2-dimensional affine subspaces, and
in particular, $|\calS_i|=4$ ($i=1,2$).
We claim that 
\be
\label{intersect_k2}
\calS_i \cap E= \{a_i,b_i\}.
\ee
Indeed, by definition, $a_i,b_i\in \calS_i$. Suppose $a_1\in \calS_2$. 
Since all EPR qubits are distinct, the inclusion $a_1\in \calS_2$ is only possible if 
$a_1=c$ or $a_1=a_2+b_2+c$. 
In both cases $c\in \aff{(E)}$, which contradicts the choice of $c$.
Thus $a_1\notin \calS_2$. Applying the same arguments to $a_2,b_1,b_2$
proves Eq.~(\ref{intersect_k2}). 

Let us first check condition CSS2 with $i=1$ (the same argument applies to $i=2$).
Choose
\[
\bar{f}(x)=\left\{ \ba{rl}
1 &\mbox{if }  x\in \calS_1\\
0 & \mbox{otherwise}\\
\ea
\right.
\]
Since $\calS_1$ is a $2$-dimensional affine subspace,
Fact~\ref{fact:affine} implies that $\bar{f}\in \calC^\perp$.
We have $\bar{f}(a_1)=\bar{f}(b_1)=1$ 
since $a_1,b_1\in \calS_1$.
From Eq.~(\ref{intersect_k2}) one gets $EX\cap \calS_1=E\cap\calS_1=\{a_1,b_1\}$.
Thus
$\bar{f}(v)=0$ for all $v\in EX\setminus \{a_1,b_1\}$, as claimed.
This proves condition CSS2.

Let us  check condition CSS1 with $i=1$ (the same argument applies to $i=2$).
We can invoke Lemma~\ref{lemma:reg} (polynomial regression) with $r=1$
and $s=4$ to show that there exists a degree-$1$ polynomial $f\, : \, \FF_2^m\to \FF_2$
such that 
\be
\label{constraints_k2}
f(a_1)=f(b_1)=1 \quad  \mbox{and} \quad f(a_2)=f(b_2)=0.
\ee
We can use condition~(2) of Lemma~\ref{lemma:reg} since $s=2^{r+1}=4$.
By definition, $f$ is a codeword of $\calC=\mathsf{RM}(1,m)$ and $f(a_1)=f(b_1)=1$.
We need to check that $f(v)=0$ for all $v\in EZ\setminus \{a_1,b_1\}$.
By definition,
\[
EZ\setminus \{a_1,b_1\} = (\calS_1\cup \calS_2)\setminus \{a_1,b_1\}=
\{a_2,\quad b_2,\quad c,\quad a_1+b_1+c,\quad a_2+b_2+c\}.
\]
We already know that $f(a_2)=f(b_2)=0$ by  Eq.~(\ref{constraints_k2}).
Since $f$ is a degree-$1$ polynomial, one has 
\[
f(a_1+b_1+c)=f(a_1)+f(b_1)+f(c) = 1 + 1+f(c) = f(c).
\]
Likewise, $f(a_2+b_2+c)=1+1+f(c)=f(c)$. 
If $f(c)=0$, then we are done. 
Suppose $f(c)=1$.
Since $c\not\in \aff{(E)}$ there is an affine subspace $S$ with co-dimension one,
which contains $\aff{(E)}$, but $c\not\in S$. Let $g$ be the linear function that is $0$
on $S$ and 1 on $\overline{S}$. Let $h=f+g$. Then $h$ on $a_1,b_1,a_2,b_2$
takes the same values as $f$, but $h(c)=0$, and we apply the above argument for 
$h$ instead of $f$. 
This proves CSS1.
\end{proof}

\subsection{\texorpdfstring{$3$}{3}-pairability}
\label{subs:3-pairable}

In the case $k=3$  we choose $m=5$ and   $\calC=\mathsf{RM}(2,5)$.
The resource state $|\calC\ra$ requires $n=32$ qubits.
We checked conditions CSS1 and CSS2 of Lemma~\ref{lemma:css} numerically
using exhaustive search over all tuples of EPR qubits and all choices
of vectors $c_1,c_2,c_3$ in the definition of subsets $X$ and~$Z$.
It was observed that for any tuple $\{a_1,b_1,a_2,b_2,a_3,b_3\}$
of EPR qubits, there exists at least one
choice of the $c$-vectors such that $X$ and $Z$ obey conditions CSS1 and CSS2.
The search space was pruned by exploiting the affine
invariance of Reed-Muller codes, see Fact~\ref{fact:symmetry}.
 Namely, choose any invertible affine map 
$\varphi\, : \, \FF_2^m\to \FF_2^m$ such that $\varphi(a_1)=0^m$ and $\varphi(b_1)=10^{m-1}$.
Fact~\ref{fact:symmetry} implies that a permutation of the $n=2^m$ qubits described by $\varphi$
is an automorphism of $\calC$. Thus this permutation of qubits
leaves the resource state $|\calC\ra$ invariant
and we can assume w.l.o.g.\ that $a_1=0^m$ and $b_1=10^{m-1}$.
We also pruned the search over the $c$-vectors by imposing a constraint
$c_1+c_2+c_3=0$
which is analogous to the constraint $c_1=c_2$ used for $k=2$.
The remaining search over $a_2,b_2,a_3,b_3$ took less than one hour on a laptop computer.
Note that the affine invariance of Reed-Muller codes also implies that 
$|\mathrm{RM}(2,m)\ra$ is $3$-pairable for all $m\ge 5$ since we can always apply an affine
map $\varphi$ as above such that $\varphi(a_i)$ and $\varphi(b_i)$ has nonzeros only on the first
$5$ bits. We note that the choice of parameters $r=2$, $m=5$ is minimal for $3$-pairability
of resource states $|\mathrm{RM}(r,m)\ra$,
as follows from a simple code distance argument.\footnote{$k$-pairability of  the CSS resource state $|\calC\ra$ requires both codes $\calC$ and $\calC^\perp$ to have minimum  distance at least $2k$ since otherwise a stabilizer of $|\calC\ra$ may have support
on a subset of the EPR qubits $a_1,b_1,\ldots,a_k,b_k$ and anti-commute with some stabilizer of the target EPR-pairs.
In particular, $3$-pairability requires $\calC$ and $\calC^\perp$ to have distance at least~$6$.
Reed-Muller codes $\calC=\mathrm{RM}(r,m)$ do not have this property for $r=1$ and $m\le 4$,
see Facts~\ref{fact:params}, \ref{fact:symmetry}. Meanwhile, the code 
$\mathrm{RM}(2,5)$ is self-dual and has distance $8$.}

\subsection{\texorpdfstring{$k$}{k}-pairability for an arbitrary $k$ (and sufficiently large $n$)}
\label{subs:k-pairable}

In this section we prove that the resource state
$|\mathsf{RM}(k-1,m)\ra$ is $k$-pairable for any $k\ge 2$ and $m\ge 3k$ (note that the number of parties can be any $n=2^m\geq 2^{3k}$). First let us
exploit the affine invariance of Reed-Muller codes
(Fact~\ref{fact:symmetry}) to convert the set of EPR
qubits $a_1,b_1,\ldots,a_k,b_k$ into a certain standard form. Choose a linear invertible map $\varphi\, : \, \FF_2^m\to \FF_2^m$ such that
$\varphi(a_i)$ and $\varphi(b_i)$ have zeros on the first $k$ bits for all $i$ (recall that we label the $n=2^m$ qubits by $m$-bit strings).
This is always possible for $m\ge 3k$.
Since the state $|\mathsf{RM}(k-1,m)\ra$ is invariant under the permutation of the $n=2^m$ qubit-labels associated with $\varphi$, we can replace
$a_i,b_i$ by $\varphi(a_i)$ and $\varphi(b_i)$. Accordingly, from now on we assume that
$a_i$ and $b_i$ have zeros on the first $k$ bits.
The linear map $\varphi$ can be computed in time $O(m^3)$ using Gaussian elimination. 
In addition, we can assume that 
\be
\label{ABoverlap}
\{a_1,b_1,\ldots,a_k,b_k\} \cap \{0^m,a_1+b_1,\ldots,a_k+b_k\} = \emptyset.
\ee
Indeed, to see this, suppose $h\in \FF_2^m$ is a vector whose first $k$ bits are zero, and 
none of the vectors $a_i+h$ or $b_i+h$ belongs to the set $\{0^m,a_1+b_1,\ldots,a_k+b_k\}$.
Using the affine invariance of Reed-Muller codes one can replace
$a_i$ and $b_i$ by $a_i+h$ and $b_i+h$. The new vectors $a_i,b_i$ obey 
the extra condition Eq.~(\ref{ABoverlap}).
The number of bad $h$s ($h$s we should not pick) is at most
\[
|\{a_1,b_1,\ldots,a_k,b_k\}| \cdot |\{0^m,a_1+b_1,\ldots,a_k+b_k\}| = 2k(k+1)
\]
(upper bounding the number of all possible differences between the two sets).
The number of $h$s we can pick from (all those vectors starting with $k$ 0s)
is at least $2^{2k}$. Now $2k(k+1) < 2^{2k}$ (which holds for all $k\ge 2$) gives the claimed property.
Hence, from now on we assume Eq.~(\ref{ABoverlap}).

We choose vectors $c_1,\ldots,c_k$ in Eq.~(\ref{Sk}) as the basis
vectors of $\FF_2^m$ such that  the $j$-th bit of $c_j$ is~1
and all other bits of $c_j$ are~0,
\be
c_j=[\underbrace{0,\ldots ,0}_{j-1} ,1, \underbrace{0,\ldots ,0}_{m-j}], \qquad j=1,\ldots,k.
\ee
Thus the $c$-vectors are supported on the first $k$ bits while all $a$- and $b$-vectors
are supported only on the last $m-k$ bits. Next we use Eqs.~(\ref{Sk},\ref{XZ})
to define the subsets of qubits $X,Z\subseteq \FF_2^m$ to be measured in the Hadamard ($X$)
and in the standard ($Z$) basis, respectively. For convenience, we restate the definitions of $X$, $Z$, and $\calS_i$ below.
\[
X = \FF_2^m \setminus  (\calS_1\cup \ldots \cup \calS_k) \quad  \mbox{and} \quad
Z = (\calS_1\cup \ldots \cup \calS_k)\setminus E,
\]
\[
\calS_i=\aff{(\{a_i,b_i,c_1,c_2,\ldots,c_k\}\setminus \{c_i\})}.
\] 
It remains to prove that $X$ and $Z$ satisfy conditions CSS1, CSS2 of Lemma~\ref{lemma:css}
with 
\[
\calC=\mathsf{RM}(k-1,m) \quad \mbox{and} \quad  \calC^\perp=\mathsf{RM}(m-k,m).
\]
Below 
we shall use the following property.
\begin{proposition}
The affine subspace $\calS_i$ is $k$-dimensional and obeys
\label{prop:intersect}
\be
\label{intersect}
\calS_i \cap E= \{a_i,b_i\}.
\ee
\end{proposition}
\begin{proof}
We have $|\calS_i|=2^k$ since all $c$-vectors
are linearly independent and have zeros on the last $m-k$ bits,
while $a_i$, $b_i$ have zeros on the first $k$ bits and $a_i\ne b_i$.
Thus $\calS_i$ is $k$-dimensional.

Let us check Eq.~(\ref{intersect}).
By definition, $\calS_i$ contains both $a_i$ and $b_i$.
Suppose $i\ne j$ and $a_j\in \calS_i$. Then 
$a_j$ is an odd linear combination of vectors
$\{a_i,b_i,c_1,c_2,\ldots,c_k\}\setminus \{c_i\}$.
Recall that the last $m-k$ bits of all $c$-vectors are zero and the first $k$ bits of all $a$- and $b$-vectors are zero.
Thus  $a_j$ must be an odd linear combination of vectors $a_i$ and $b_i$ only.
This is only possible if $a_j=a_i$ or $a_j=b_i$. However, we assumed that all EPR
qubits $a_1,b_1,\ldots,a_k,b_k$ are distinct. Thus $a_j\notin \calS_i$.
The same argument shows that $b_j\notin \calS_i$.
\end{proof}

First let us check condition CSS2 with $i=1$ (the same argument works for any $i$).
Choose a function $\bar{f}\, : \, \FF_2^m\to \FF_2$ as
\[
\bar{f}(x)=\left\{ \ba{rl}
1 &\mbox{if }  x\in \calS_1\\
0 & \mbox{otherwise}\\
\ea
\right.
\]
Since $\calS_1$ is a $k$-dimensional affine subspace,
Fact~\ref{fact:affine} implies that $\bar{f}\in \calC^\perp$.
We have $\bar{f}(a_1)=\bar{f}(b_1)=1$ 
since $a_1,b_1\in \calS_1$.
From Eq.~(\ref{intersect}) one gets $EX\cap \calS_1=E\cap\calS_1=\{a_1,b_1\}$.
Thus
$\bar{f}(v)=0$ for all $v\in EX\setminus \{a_1,b_1\}$, as claimed.
This proves condition CSS2.

Checking  condition CSS1 requires more
technical work, and we strongly encourage the reader to first study the proof 
of a quite general special case in Appendix \ref{app:example},
which is much simpler. 

As before, we can focus on the
case $i=1$ (the same argument works for any $i$).
Then condition CSS1 is equivalent to
the existence of a degree-$(k-1)$ polynomial
$f\, : \, \FF_2^m \to \FF_2$ such that 
\be
\label{desired}
f(a_1)=f(b_1)=1 \quad \mbox{and} \quad 
f(x)=0 \quad \mbox{for all} \quad x\in (\calS_1\cup \ldots \cup \calS_k) \setminus \{a_1,b_1\}.
\ee
Any degree-$(k-1)$ polynomial $f\, : \, \FF_2^m\to \FF_2$
can be written as
\be
\label{f(x)}
f(x)=\sum_{T\subsetneq [k]} f_T(x) \prod_{j\in T} x_j
\ee
where $f_T\, : \, \FF_2^{m} \to \FF_2$ is some polynomial of degree $k-1-|T|$
that depends only on the variables $x_{k+1},\ldots,x_{m}$.
It remains to choose the polynomials $f_T$ with $0\le |T|\le k-1$.
We shall use induction on $|T|$ starting with $T=\emptyset$.
At each induction step we shall use polynomial regression (Lemma~\ref{lemma:reg})
to argue that the desired polynomial $f_T$ exists 
(there is no need to construct $f_T$ explicitly). 
Given $\emptyset \neq T\subsetneq [k]$, define the following set of $m$-bit strings
\be
\label{e(S)}
e(T) = \left\{ \ba{rl}
\{ a_i,b_i\, \mid \, i\in [k]\setminus T\}
& \mbox{if $|T|$ is even},\\
\{0^m\} \cup \{a_i+b_i\, \mid \, i\in [k]\setminus T\}
& \mbox{if $|T|$ is odd}.\\
\ea
\right.
\ee

\begin{proposition}
\label{prop:induction}
Consider 
 a function $f(x)$ of the form Eq.~(\ref{f(x)}),
where the $f_T(x)$ are some functions that depend only on the variables $x_{k+1},\ldots,x_m$.
Then $f(x)$
satisfies condition Eq.~(\ref{desired}) iff:
\be
\label{induction_1}
f_\emptyset(a_1)=f_\emptyset(b_1)=1, \quad  \quad f_\emptyset(a_i)=f_\emptyset(b_i)=0 \quad
\mbox{for $2\le i\le k$},
\ee
and for every $\emptyset \neq T\subsetneq [k]$ and every $x\in e(T)$:
\be
\label{induction_2}
\sum_{U\subseteq T} f_U(x) = 0 \;\;\;\;\;\;\; \left({\rm equivalently,}\;
f_T(x)=\sum_{U\subsetneq T} f_U(x) \right)
\ee
\end{proposition}
\begin{proof}
CSS1 requires that $f$ takes given values on $S = \calS_1\cup \ldots \cup \calS_k$, namely 
$f(a_{1})= f(b_{1}) = 1$, and $f$ must be zero on the rest of $S$. When we plug these input values
into Eq.~(\ref{f(x)}), it is a straightforward calculation to see
that Eq.~(\ref{induction_1}) exactly says that  $f(a_{i})= f(b_{i}) = 1$ if and only if $i=1$,
and Eq.~(\ref{induction_2}) exactly says that  $f(x)=0$ for all $x\in S\setminus E$.
More precisely, Eq.~(\ref{induction_2}) for some $\emptyset \neq T\subsetneq [k]$ and $x\in e(T)$
expresses exactly $f(y)=0$ for $y=\chi_{T}+x\in S$, where $\chi_{T}$
is the characteristic vector of~$T$. Eq.~(\ref{e(S)}) was engineered so that 
the above $y$ ranges over all of $S\setminus E$.
\end{proof}

In the rest of the section we concentrate our efforts on finding a family $\{f_T\}_{T\subsetneq [k]}$
of polynomials that satisfy Proposition 
\ref{prop:induction}.
Below we focus on the case when $k$ is even.
The analysis for odd $k$ requires only a few minor modifications,
as detailed in  Appendix~\ref{app:odd_k}.

For fixed $k, a_{i}s$ and $b_{i}s$ we recursively (in the increasing size of $|T|$) construct
$f_T$, and inductively show that
$f_T$ has the desired properties.
To enable induction,  we supplement Eqs.~(\ref{induction_1},\ref{induction_2})
with a few extra conditions on the polynomials $f_T$.
Let $\ell\ge 0$ be an integer. 
We will say that a family of polynomials $f_T\, : \, \FF_2^m\to \FF_2$
labeled by subsets $T\subsetneq [k]$ with $|T|\le \ell$
is \emph{valid} if all conditions stated below are satisfied for $|T|\le \ell$:
\begin{enumerate}
\item[\bf I1:] $f_T$ depends only on the variables $x_{k+1},\ldots,x_m$
\item[\bf I2:]  $f_T$ has degree $k-1-|T|$
\item[\bf I3:]  $f_\emptyset(a_1)=f_\emptyset(b_1)=1$ and $f_\emptyset(a_i)=f_\emptyset(b_i)=0$ for $2\le i\le k$
\item[\bf I4:]  $f_T(x)=\sum_{U\subsetneq T} f_U(x)$ for every non-empty set $T\subsetneq [k]$ and every $x\in e(T)$
\item[\bf I5:]  $f_T\equiv 0$ if $|T|$ is odd and $|T|\le k-3$
\item[\bf I6:]  $f_T(0^m)=f_T(a_i+b_i)=0$ if $i\notin T$ and $|T|\le k-4$
\item[\bf I7:] $f_T(a_i)=f_T(b_i)$ if $i\notin T$
and $|T|\le k-2$
\end{enumerate}
Here  (I1,I3,I4) are the conditions stated in Proposition~\ref{prop:induction}
and (I2) ensures that a polynomial $f(x)$ constructed 
from the family $\{f_T\}_{T\subsetneq [k]}$ according to Eq.~(\ref{f(x)}) has degree $k-1$.
Thus (I1,I2,I3,I4) alone imply CSS1.
We shall use induction on $\ell$ to prove that a valid family of polynomials 
exists for all $\ell\le k-1$.   The extra conditions
(I5,I6,I7) facilitate analysis of the induction step. 

\medskip

The base case of the induction is $\ell=0$.
Then a valid family is a single polynomial
$f_\emptyset$.  
Condition (I2) demands that 
$f_\emptyset$ has degree $d_T=k-1$.
 Conditions (I4) and (I5) can be skipped for $T=\emptyset$. Condition (I7) follows trivially from (I3). It remains to check (I3,I6)
Note that conditions (I3) and (I6) with $T=\emptyset$
are imposed at disjoint set of points,
see Eq.~(\ref{ABoverlap}).
Thus (I3) and (I6) are consistent. 
We fix the value of $f_\emptyset$ at $s=2k$ points in (I3)
if $k\le 3$, and at $s=3k+1$ points in (I3,I6) if $k\ge 4$.
We can show that the desired polynomial $f_\emptyset$ exists
using Lemma~\ref{lemma:reg}. Below we always
apply the lemma to polynomials satisfying condition (I1). This is justified since the first $k$ bits
of all vectors $a_i$ and $b_i$ are zero. 
If $k\le 3$ then
 we have $s=2k\le 2^{d_T+1}=2^k$.
Thus we can use part~(2) of Lemma~\ref{lemma:reg}.
The extra condition $\sum_{i=1}^s g_i=0$ of the lemma
is satisfied since (I3) fixes the value of $f_\emptyset$
to $1$ at an even number of points.
If $k\ge 4$ then
we have $s=3k+1<2^{d_T+1}=2^k$ and thus 
we can use 
part~(1) of
Lemma~\ref{lemma:reg}.

\medskip

We shall now prove the induction step.
Suppose we have already constructed a valid family of polynomials $f_T$
with $|T|\le \ell-1$. 
Consider a subset $T\subsetneq [k]$ with $|T|=\ell$ such that $1\le \ell\le k-2$. The case $\ell=k-1$ will be considered afterwards.

Suppose $|T|$ is odd. Since $k$ is even
and $\ell\le k-2$, this is only possible if $|T|\le k-3$.
We set $f_T\equiv 0$ to satisfy (I5).
Then  conditions (I2), (I6), (I7) 
are satisfied automatically.
Condition (I3) can be skipped since $T\ne \emptyset$.
Condition (I4) with $f_T\equiv 0$ demands that  $\sum_{U\subsetneq T} f_U(x)=0$ for any $x\in e(T)$.
Since $|T|$ is odd,  $e(T)$ is the set of points $0^m$ and $a_i+b_i$ with $i\notin T$.
Each term $f_U(x)$ with odd $|U|$ vanishes due to (I5). Each term $f_U(x)$
with even $|U|$ vanishes for $x\in e(T)$ due to (I6) since $|U|\le |T|-1\le k-4$.
Thus choosing $f_T\equiv 0$ satisfies (I4).  

Suppose $|T|$ is even. Condition (I2)
demands that $f_T$ has degree $d_T=k-1-|T|$.
Conditions (I3), (I5) can be skipped
since $|T|$ is even and $T\ne \emptyset$.
We claim that (I7) follows from (I4).
Indeed, we have $x\in e(T)$ iff
$x=a_i$ or $x=b_i$ with $i\notin T$.
We have $f_T(x)=\sum_{U\subsetneq T} f_U(x)$
due to (I4). If $|U|$ is odd then
$f_U\equiv 0$ due to (I5) since $|U|\le |T|-1=\ell-1\le k-3$. If $|U|$ is even then $|U|\le |T|-2\le k-4$.
Thus $f_U(a_i)=f_U(b_i)$ since (I7) holds for $f_U$.
This implies $f_T(a_i)=f_T(b_i)$ for any $i\notin T$, as claimed in (I7).
It remains to check (I4,I6).
We fix the value of $f_T$ at $s=|e(T)|=2(k-|T|)$ points in (I4) 
if $|T|\ge k-3$. We fix the value of $f_T$ at $s=|e(T)|+k-|T|+1 = 3(k-|T|)+1$
points in (I4,I6) if $|T|\le k-4$.
Consider first the case $|T|\ge k-3$.
Then $|T|=k-2$ since $k$ and $|T|$ are even.
Thus $s=2(k-|T|)=4$ and $2^{d_T+1}=2^{k-|T|}=4$.
We can use part~(2) of Lemma~\ref{lemma:reg}
to construct $f_T$.
The extra condition $\sum_{i=1}^s g_i=0$
of the lemma is equivalent to 
\be
\label{extra_condition_1}
\sum_{x\in e(T)} \sum_{U\subsetneq T} f_U(x)=0.
\ee
Each term $f_U(x)$ with odd $|U|$ vanishes due to (I5)
since $|U|\le |T|-1=k-3$. 
Each term $f_U(x)$ with even $|U|$ obeys
$f_U(a_i)=f_U(b_i)$ for $i\notin T$
since $f_U$ obeys (I7). 
Thus we have $\sum_{x\in e(T)} f_U(x)=0$
for any $U\subsetneq T$ which implies Eq.~(\ref{extra_condition_1}). Thus the desired polynomial $f_T$ exists by part~(2)  of
Lemma~\ref{lemma:reg}, that is,
we have checked (I4) in the case $|T|\ge k-3$.
Condition (I6) can be skipped in this case.
In the remaining case, $|T|\le k-4$,
we can use part~(1) of Lemma~\ref{lemma:reg}
since $s=3(k-|T|)+1<2^{d_T+1}=2^{k-|T|}$
for $|T|\le k-4$.
Thus conditions (I4,I6) are satisfied. 

It remains to prove the induction step for $\ell=k-1$.
Suppose we have already constructed a valid family of polynomials $f_T$
with $|T|\le k-2$. 
Consider a subset $T\subsetneq [k]$ with $|T|=k-1$.
Since we assumed that $k$ is even, $|T|$ is odd.
Condition (I2) demands that $f_T$
has  degree $k-1-|T|=0$, that is, $f_T$ is
a constant function.
We can skip conditions (I3,I5,I6,I7)
since none of them applies if $|T|=k-1$.
It remains to check (I4).
Note that $e(T)=\{0^m,a_i+b_i\}$ for some $i\in [k]$ such that $T=[k]\setminus \{i\}$.
Condition (I4) fixes the value of $f_T(x)$ at
$x=0^m$ and at $x=a_i+b_i$.
Since we want $f_T$ to be a constant function, it suffices to check that the desired
values $f_T(0^m)$ and $f_T(a_i+b_i)$
are the same.
Substituting the desired values from (I4),
we have to check that 
\be
\label{extra_condition_2}
\sum_{U\subsetneq T} f_U(0^m) + f_U(a_i+b_i)=0.
\ee
The sum contains terms with $|U|\le |T|-1= k-2$.
All terms $f_U$ with odd $|U|$ must have
$|U|\le k-3$ since $k$ is even. Such terms vanish due to (I5).
All terms $f_U$ with even $|U|\le k-4$ vanish due to (I6).
Thus we can restrict the sum Eq.~(\ref{extra_condition_2}) to terms
with $|U|=k-2$.
However, $f_U$ is a degree-1 polynomial if $|U|=k-2$ due to (I2). 
Thus $f_U(0^m)+f_U(a_i+b_i) = f_U(a_i)+f_U(b_i)$
and Eq.~(\ref{extra_condition_2}) is equivalent to
\be
\label{extra_condition_3}
\sum_{\substack{U\subsetneq T\\ |U|=k-2\\}}\;  f_U(a_i) + f_U(b_i)=0.
\ee
Since $f_U$ obeys (I7), we have $f_U(a_i)=f_U(b_i)$,
which implies Eq.~(\ref{extra_condition_3}).
We have now verified (I4). This completes the proof
of the induction step. 

Accordingly, having shown that both conditions CSS1 and  CSS1 of Lemma~\ref{lemma:css} are satisfied, we can now conclude that the resource state $|\mathsf{RM}(k-1,m)\ra$ is $k$-pairable.

\subsection{10-qubit 2-pairable example}
\label{subs:10qubit}

The 2-pairable state of Section~\ref{subs:2-pairable} used $n=16$ qubits. Extending 2-pairability to states with fewer qubits would be good. Here we give
a $10$-qubit example and describe Pauli measurements generating $k=2$ EPR-pairs
for all choices of such pairs (modulo certain symmetries).

We choose the resource state $|\psi\ra$ as the graph
state associated with the 10-vertex ``wheel graph'' shown in Figure~\ref{fig:graph_state}:
\be
|\psi\ra = \prod_{(i,j)\in E} \mathsf{CZ}_{i,j} |+\ra^{\otimes 10}.
\ee
Here $E$ is the set of graph edges and $\mathsf{CZ}$ is the controlled-$Z$ gate.

\setlength{\intextsep}{10pt}%
\begin{figure}[hbt]
\centerline{
\includegraphics[height=3cm]{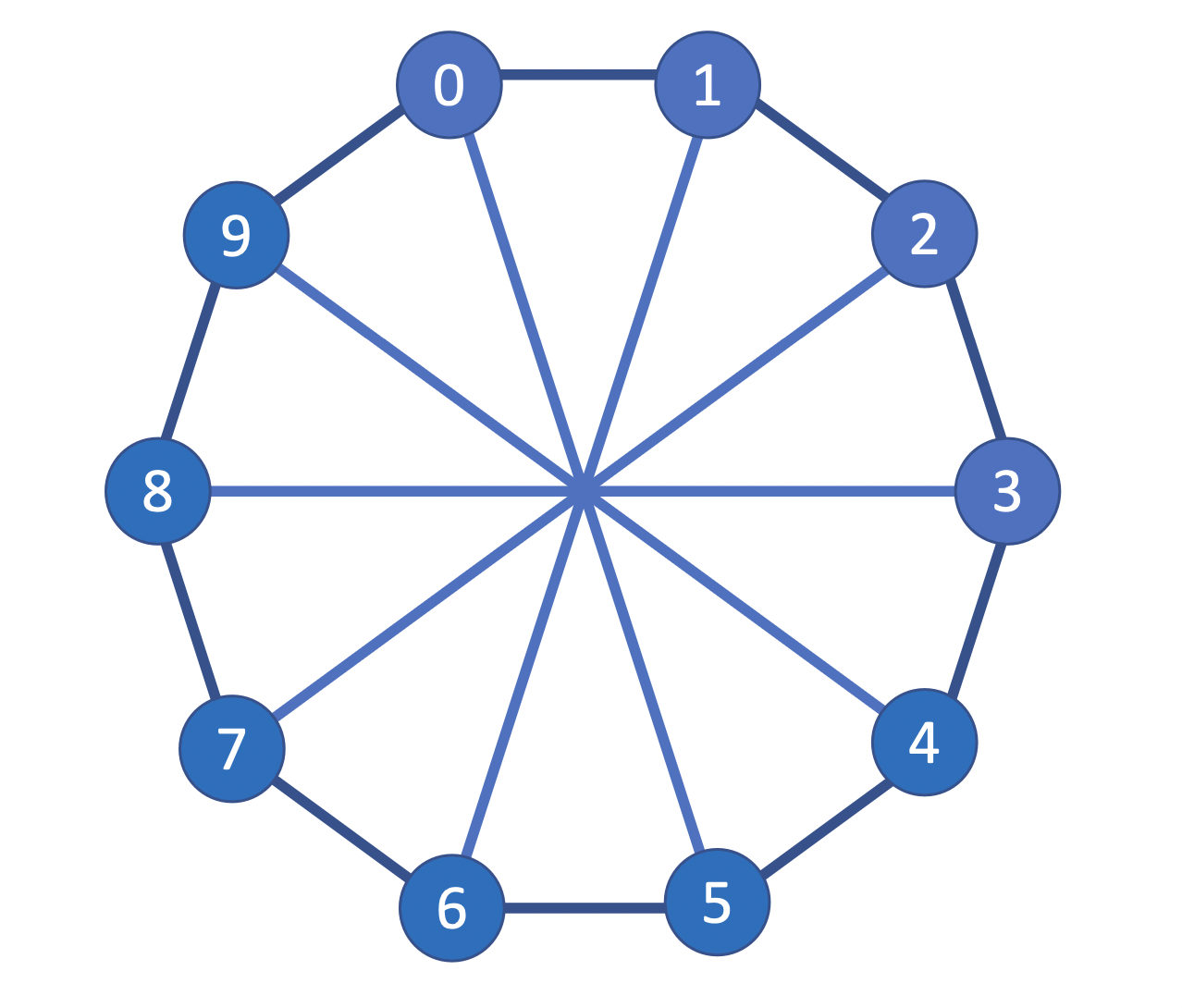}}
  \caption{$10$-vertex ``wheel graph''. The corresponding $10$-qubit graph state is $2$-pairable with one qubit per party (to avoid confusion: the center of the picture is not an 11th vertex).}
  \label{fig:graph_state}
\end{figure}

The number of ways to choose two EPR-pairs
$\{a_1,b_1\}$ and $\{a_2,b_2\}$ is $3\binom{n}{4} =630$ for $n=10$ qubits.
However, the number of cases we need to consider can be reduced by noting that 
the graph state~$|\psi\ra$ is invariant under certain permutations of qubits
and local Clifford operations. Indeed, if $W_\varphi$ is a permutation of $n$ qubits
(considered as a unitary operator) and $C$ is a product of single-qubit Clifford gates
such that $W_\varphi |\psi\ra=C|\psi\ra$, then an LOCC protocol generating 
EPR-pairs $\{a_1,b_1\}$ and $\{a_2,b_2\}$ can be easily converted into one generating EPR-pairs $\{\varphi(a_1),\varphi(b_1)\}$
and $\{\varphi(a_2),\varphi(b_2)\}$. This conversion requires only relabeling of qubits
and local basis changes. 

Suppose qubits are labeled by elements of the cyclic group
$\ZZ_{10} = \{0,1,\ldots,9\}$.
Clearly, $|\psi\ra$ is invariant under the cyclic shift of qubits,
$j\to j+1$ and inversion $j\to -j$.
Here and below qubit indexes are computed modulo $10$. 
Consider a permutation $\varphi\, : \, \ZZ_{10} \to \ZZ_{10}$ such that 
$\varphi(j)=3j$. Let $W_\varphi$ be the $10$-qubit unitary that 
implements the permutation $\varphi$.
We claim that 
\be
\label{symmetry_n10}
W_\varphi |\psi\ra = H^{\otimes 10} |\psi\ra,
\ee
where $H$ is the Hadamard gate.
Indeed, it is known~\cite{raussendorf2003measurement} that the graph state $|\psi\ra$ has stabilizers 
\[
S_j =\sigma^x_j \prod_{i\, : \, (i,j)\in E}\; \sigma^z_i=
\sigma^x_j \sigma^z_{j-1} \sigma^z_{j+1} \sigma^z_{j+5}, \qquad j\in \ZZ_{10}.
\]
Thus $|\psi\ra$ is also stabilized by 
\[
S_j' := S_{j-3} S_{j+3} S_{j+5} = \sigma^z_j \sigma^x_{j-3} \sigma^x_{j+3} \sigma^x_{j+5}.
\]
It follows that $H^{\otimes 10} |\psi\ra$ is stabilized by 
\[
S_j'' :=H^{\otimes 10} S_j' H^{\otimes 10}= \sigma^x_j \sigma^z_{j-3} \sigma^z_{j+3} \sigma^z_{j+5} = W_\varphi S_i W_\varphi^\dag,
\]
where $i=\varphi^{-1}(j)$.
Thus $W_\varphi^\dag H^{\otimes 10}|\psi\ra$ is stabilized by $S_i$
for all $i\in \ZZ_{10}$, which implies Eq.~(\ref{symmetry_n10}).

Since $|\psi\ra$ is also invariant under the cyclic shift of qubits, we can assume w.l.o.g.\
that $a_1=0$. The permutation $\varphi$ maps $0$ to $0$ while any qubit $b_1\in \ZZ_{10}\setminus \{0\}$
can be mapped to either $1$, or $2$, or $5$ by repeated applications of $\varphi$.
Thus we can assume w.l.o.g.\ that $a_1=0$ and $b_1\in \{1,2,5\}$.

For each of the remaining choices of EPR-pairs
we numerically examined all $3^6$
Pauli measurement bases on  qubits $\ZZ_{10}\setminus \{a_1,b_1,a_2,b_2\}$
and computed the final post-measurement state
of qubits $a_1,b_1,a_2,b_2$ using the standard stabilizer formalism. To test whether the final
state is locally equivalent to the desired EPR-pairs,
we checked whether the entanglement entropies of the final state obey
$S(a_i)=S(b_i)=1$ and $S(a_ib_i)=0$ for $i=1,2$.
The entanglement entropy of a stabilizer state can
be extracted from its  tableaux as described in~\cite{Fattal2004EntanglementIT}.
Any two-qubit stabilizer state of qubits $a_i,b_i$ satisfying $S(a_i)=S(b_i)=1$ and $S(a_ib_i)=0$ has to be maximally entangled and thus equivalent to the EPR-pair modulo single-qubit Clifford gates. We 
found  a Pauli basis generating maximally-entangled states on qubits
$\{a_1,b_1\}$ and $\{a_2,b_2\}$
 in all considered cases, see Figure~\ref{fig:measurements_n10}.
We also observed that the graph state $|\psi\ra$ is not $2$-pairable
if the Pauli bases are restricted to $\sigma^x$ and $\sigma^z$ only.

\setlength{\intextsep}{10pt}%
\begin{figure}[ht]
\centerline{\includegraphics[height=6cm]{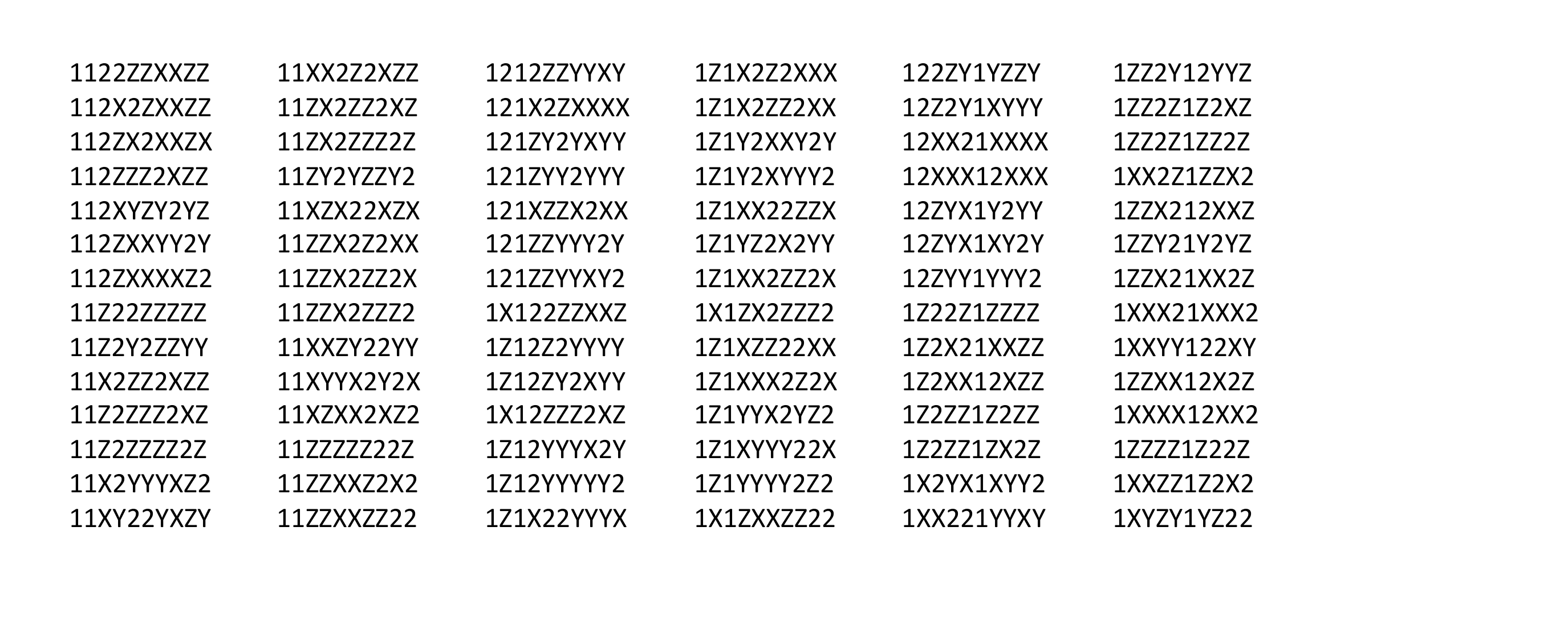}}
\caption{Measurement patterns for the $10$-qubit $2$-pairable
resource state associated with the
``wheel graph''.
Here '1' and '2' stand for the EPR qubits $\{a_1,b_1\}$ and $\{a_2,b_2\}$
respectively.
A qubit labeled by 'X', 'Y', or 'Z' is measured in the Pauli basis
$\sigma^x$, $\sigma^y$, and $\sigma^z$ respectively.
Here we only consider the case $a_1=0$ and $b_1\in \{1,2,5\}$.
All other cases can be obtained by a permutation of qubits
that leaves the resource state invariant (modulo a bitwise Hadamard).
}
\label{fig:measurements_n10}
\end{figure}

 We verified numerically that no stabilizer state with $n<10$ qubits is $2$-pairable using LOCC protocols based on Pauli measurements, by checking all possible 9-qubit graph states as listed in~\cite{adcock2020mapping}. 
The code is available at
\href{https://github.com/yashsharma25/generating-k-epr-pairs}{https://github.com/yashsharma25/generating-k-epr-pairs}

\section{Obstructions for complete pairings ($k=n/2$)}\label{sec:obstructionskisn/2}

Now we turn from constructions to proving limitations on all possible $k$-pairable resource states.
Let $n$ be the number of parties as in the previous sections. Since we are talking about complete pairings in this section, we assume here that $n$ is divisible by 2. For a pairing 
\[
\pi = \{\{a_{1},b_{1}\},\ldots,\{a_{n/2},b_{n/2}\}\} \;\;\; {\rm with} \;\;\; \cup \pi = [n]
\]
the tensor product $|\pi\rangle$ of the $n/2$ EPR-pairs,
\[
|\pi\rangle = \bigotimes_{i=1}^{n/2}\,  |\Phi^+\rangle_{a_{i},b_{i}} 
\]
is a state on $n$ qubits, where  by definition $|\Phi^+\rangle_{a,b} = \frac{1}{\sqrt{2}}\left(|0_{a}0_{b}\rangle +|1_{a}1_{b}\rangle\right)$.

For our first type of lower bounds we assume that the $n$ parties want to achieve all possible {\it complete} pairings on $[n]$. Then we find a super-constant lower bound on the required number~$m$ of qubits per party:

\begin{theorem}\label{thm:obstr1}
Suppose $\ket{\psi}$ is a fixed state of $nm$ qubits shared by $n$ parties such that each party holds
$m$ qubits of $\ket{\psi}$. Suppose that for any pairing $\pi$ of $n$ qubits a transformation $|\psi\rangle\rightarrow |\pi\rangle\otimes \ket{w_{\pi}}$ is realizable
by an LOCC protocol 
such that at the end of the protocol the $i$-th qubit of $\ket{\pi}$ belongs to the $i$-th
party for all $i$, and $\ket{w_{\pi}}$
is an arbitrary state on the qubits not 
belonging to $\ket{\pi}$.\footnote{Without loss of generality we may assume  $\ket{w_{\pi}}=\ket{0^{n(m-1)}}$.} Then
\[
m = \Omega(\log\log n).
\]
\end{theorem}

The proof of this theorem is going to be a dimension calculation, but 
with a twist. Given a starting state $\ket{\psi}$ we estimate the dimension of the space that contains all those states that can be obtained (with positive probability) from 
$\ket{\psi}$ by an LOCC protocol.
We want to compare this with the dimension of the space induced by all possible states
that should arise as output, where we let the input range 
over all possible pairings. This by itself, however, will 
not yield the desired lower bound. The mathematical idea is that rather than representing
each state by itself, we represent it by its $r^{\rm th}$ tensor power, where $r$ will
be carefully set in the magnitude of $\Theta(\log n)$.
Let
\[
{\cal L}_{r}= {\rm span}(\; |\pi\rangle^{\otimes r}\mid  \mbox{$\pi$ is an $n$-qubit pairing}\, )
\]
be the linear space induced by the $r^{\rm th}$ tensor powers of all possible 
output states.
Before stating a lower bound on $\dim ({\cal L}_{r})$ we prove a lemma:
\begin{lemma}\label{lem:in}
Let $\pi$ and $\rho$ be two pairings. Then $\langle \pi |\rho\rangle = 2^{\mu-n/2}$, where $\mu$ is the number of cycles in $\pi\cup \rho$, as a graph on 
vertex set $[n]$.
\end{lemma}

\begin{proof}
Note that the graph $\pi\cup \rho$ (the union of the two perfect matchings $\pi$ and $\rho$ on the same vertex set $[n]$) is a collection of cycles.
We have
\[
\langle \pi |\rho \rangle = \sum_{x\in \Lambda} \frac{1}{2^{n/2}}
\]
where $\Lambda\subseteq \{0,1\}^{n}$ is the set of binary strings, corresponding to the vertex labeling~$\lambda$ of the graph $\pi\cup \rho$
such that for every $\{a_{i},b_{i}\}\in \pi$ we have $\lambda(a_{i})= \lambda(b_{i})$, and
for every $\{a'_{i},b'_i\}\in \rho$ we have $\lambda(a'_{i})= \lambda(b'_{i})$. In other words, the labeling $\lambda$ must be 
constant on each connected component of $\pi\cup \rho$. Therefore,
\[
|\Lambda| \; = \; 2^{\mbox{\# of connected components of $\pi\cup \rho$}} \; = \; 2^{\mbox{\# of cycles in $\pi\cup \rho$}}.
\]
\end{proof}

\begin{lemma}\label{lem:L}
There exist constants $C, C'>0$ such that 
$\dim ( {\cal L}_{C' + \log_{2} n} )\ge {n}^{n/4} \cdot 2^{-Cn}$
\end{lemma}
\begin{proof}
Let $A = [n/2]$, $B = [n/2+1,n]$. Let 
\begin{eqnarray*}
V & = & \left\{\, \pi\mid  \mbox{$\pi$ is an $n$-qubit pairing with $a_{i}\in A$, $b_{i}\in B$ for $1\le i \le n/2$\,}\right\} 
\end{eqnarray*}
We define an undirected graph $G$ on $V$ by
\begin{eqnarray*}
V(G) & = & V \\
E(G) & = & \{ (\pi,\rho) \in V^{2} \mid \, \pi\neq \rho \; {\rm and} \; \pi\cup \rho \; \mbox{consists of at least $n/4$ cycles}\}
\end{eqnarray*}
Let us view each $\pi\in V$ as a 1-1 map from $A$ to $B$.
Then $(\pi,\rho) \in E(G)$ if and only if $\pi\rho^{-1}$ is a permutation on $A$ with at least $n/4$ cycles.
If we fix $\pi$, then as $\rho$ varies, $\pi\rho^{-1}$ runs through all permutations of $A = [n/2]$. Thus, 
every vertex of $G$ has degree $D$, where $D$ is the number of permutations of $[n/2]$ having at least $n/4$ cycles.
Let $c(n,\ell)$ be the unsigned Stirling numbers of the first kind. It is known that $c(n,\ell)$ is exactly the number of 
permutations of $n$ elements with $\ell$ disjoint cycles. Thus
\[
D = \sum_{\ell=0}^{n/4} c(n/2,n/4 + \ell)
\]
It is also known that
\[
c(n,n-\ell) \; =  \; \sum_{0\le i_{1} < i_{2} < \cdots < i_{\ell} < n}  i_{1}i_{2}\cdots i_{\ell}  
\]
The right-hand side above is at most
\[
 \binom{n}{\ell} \cdot (n-\ell)\cdots (n-1) \; \le \; \frac{ 2^n}{(n-\ell-1)!}   (n-1)!                                           
\]
Therefore:
\[
D \; = \; \sum_{\ell=0}^{n/4} c(n/2,n/4 + \ell) \; \le \; (n/2)! \; \cdot\;  
\frac{2}{n}\sum_{\ell=0}^{n/4} \,  \frac{2^{n/2}}{(n/4 + \ell-1)!} \; \le \; (n/2)! \frac{2^{n/2}}{(n/4 -1)!} -1.
\]
Then, 
\[
\frac{|V|}{D+1} = \frac{(n/2)!}{D+1}  \ge \frac{(n/4 -1)!}{2^{n/2}}, 
\]
implying the existence of an independent set in $G$ of size at least $(n/4 -1)! / 2^{n/2}$
by a well-known greedy argument: pick a vertex to add to the independent set, remove it and its $\leq D$ neighbors, and continue with the remaining graph. Using Stirling's formula to estimate the factorials, there is a $C$ ($\approx 1 + \frac{1}{4} \log_{2} e$, when $n\rightarrow \infty$)
such that $(n/4 -1)! / 2^{n/2} \ge {n}^{n/4} \cdot 2^{-Cn}$.  Let $I$ be an independent set in $G$ of size ${n}^{n/4} \cdot 2^{-Cn}$ (we ignore rounding to an integer for simplicity). Define the linear space
\begin{eqnarray*}
{\cal L}_{I,r} & = {\rm span}(\; |\pi\rangle^{\otimes r}\mid  \pi\in I \,).
\end{eqnarray*}
From now on we set $r= C' + \log_{2} n$, where $C' = 1 +4C$ ($\approx 5 + \log_{2} e$, when $n\rightarrow \infty$). 
It is enough to show that we have $\dim ( {\cal L}_{I,r} ) = |I|$ ($= {n}^{n/4} \cdot 2^{-C}$), since ${\cal L}_{I,r}\le {\cal L}_{r}$.
Let us define the $|I|\times|I|$ Gram matrix $\mathbb{G}$ of the $\{ |\pi\rangle^{\otimes r}\}_{\pi\in I}$ system:
\[
\mathbb{G}_{\pi,\rho} \; = \; \langle \pi |\rho \rangle^r \;\;\;\;\;\;\;\;\   \pi , \rho \in I
\]
For $\dim ( {\cal L}_{I,r} ) = |I|$ it is sufficient to show that $\mathbb{G}$ has full rank.

 Gershgorin's circle theorem implies, for any (complex) square matrix $A$:  if $R_{i}=\sum_{j:\; j\neq i} |A_{i,j}| < |A_{i,i}|$ for 
all indices $i$, then $A$ has full rank. 

In order to apply this theorem, we compute for the $\pi$-row of our matrix $\mathbb{G}$:
\[
\sum_{\rho\in I:\rho\neq\pi} |\mathbb{G}_{\pi,\rho} | \; \le \; \sum_{\rho\in I:\rho\neq\pi} 2^{-nr/4} \; 
\leq \; 2^{-nr/4} |I| \; = 
\; 2^{-n(\underbrace{\scriptstyle 1 + 4C+ \log_{2} n}_{r})/4} \cdot \underbrace{{n}^{n/4} \cdot \, 2^{-Cn}}_{\rm size \; of \; {\it I}}  < 1
\]
where we used Lemma \ref{lem:in}. Since in addition, 
$\mathbb{G}_{\pi,\pi}=1$ for all $\pi\in I$,  Gershgorin's circle theorem implies that our matrix $\mathbb{G}$ has full rank.
Hence $\dim ( {\cal L}_{I,r} ) = |I| = {n}^{n/4} \cdot 2^{-Cn}$.
\end{proof}

Lemma~\ref{lem:L} says that for any fixed state  $|\psi\rangle$, the possible outputs (over all input pairings $\pi$), 
when taking their $r^{\rm th}$ tensor power with $r = \Theta(1) + \log n$, should span a space of dimension $\geq n^{\Theta(\log n)}$.

This number we have to compare with the dimension of the span of $r^{\rm th}$ tensor powers of possible states that can be produced \emph{by an LOCC protocol} from $\ket{\psi}$. Although LOCC protocols may use unlimited classical communication, they cannot create new entanglement, so all entanglement in their final state
is a \emph{local linear transformation} of the entanglement that already existed in the starting state $\ket{\psi}$.

When each party only possesses $m$ qubits, where $m$ is very small, the variety 
of states that an LOCC protocol can produce from $\ket{\psi}$ is limited in the way we describe below.

To capture this limitation, notice that any LOCC protocol can be described by a completely positive trace-preserving (CPTP) map, with \emph{separable}
Kraus operators. It follows that for any pairing $\pi$ there exists a product Kraus operator
\begin{equation}\label{eq:M1}
K^{\pi} =
K_{1}^{\pi} \otimes K_{2}^{\pi} \otimes \cdots \otimes K_{n}^{\pi}
\end{equation}
such that $K_{i}^{\pi}$
maps $m$ qubits to one qubit for all $1\le i\le n$ and
for all pairings $\pi$:
\begin{equation}\label{eq:M2}
K^{\pi} |\psi\rangle
= c_{\pi}|\pi\rangle 
\;\;\;\;\;
{\rm for}\;{\rm some}\;\;
c_{\pi}\neq 0
\end{equation}
Define:
\[
{\cal M}_{r} = 
{\rm span}(\;
(K|\psi\rangle)^{\otimes r}\;\mid\;
K =
K_{1} \otimes K_{2} \otimes \cdots \otimes K_{n})
\]
where the $K_{i}$ are arbitrary
operators mapping $m$ qubits to one qubit ($K_i$ may depend on $i$), and 
$|\psi\rangle$ is our fixed starting state. From Equations 
(\ref{eq:M2}) and (\ref{eq:M1}) we get the subspace inclusion
\[
{\cal L}_{r} \le {\cal M}_{r}
\]
and hence, using Lemma~\ref{lem:L}, for some $C, C'>0$ we have
$\dim ( {\cal M}_{C' + \log_{2} n} ) \ge {n}^{n/4} \cdot 2^{-C}$. 
However, when 
$m=o(\log\log n)$ this cannot be the case because of the following upper bound:

\begin{lemma}\label{lem:dimMr}
$\displaystyle
\dim \left({\cal M}_r \right)
\le \binom{2^{m+1} + r -1}{2^{m+1} -1}^{n}.
$
\end{lemma}

\begin{proof}
Linear operators $K_{i}$ that map 
$m$ qubits to one qubit can be considered as vectors in a complex
space of dimension $D = 2^{m+1}$ (use the vectorized form of operators).
Crucially, the $r$-fold tensor products
$K_{i}^{\otimes r}$ live in the \emph{symmetric subspace} of
$\left(\mathbb{C}^{D}\right)^{\otimes r}$,
which has dimension 
\begin{equation}\label{eq:symm}
{D + r -1 \choose D -1}
\end{equation}
(this is where the big saving occurs: without the information that the vector is in the symmetric subspace, we would have to calculate with $D^r$ instead of the above expression, and would get only a trivial bound).
It follows that operators of the form
$K^{\otimes r} =
K_{1}^{\otimes r} \otimes \cdots \otimes K_{n}^{\otimes r}$
span a linear space of operators
with dimension at most
\begin{equation}\label{eq:symmn}
{D + r -1 \choose D -1}^{n}
\end{equation}
Thus states of the form 
$\left(K|\psi\rangle\right)^{\otimes r}
= K^{\otimes r}|\psi\rangle^{\otimes r}$
with a fixed $|\psi\rangle$  span a linear space with
dimension upper bounded by Equation~(\ref{eq:symmn}).
Substituting $D = 2^{m+1}$, 
one gets the statement of the lemma.
\end{proof}

It is now an easy calculation 
to show that with $r = \Theta(\log n)$, the above lemma together with Lemma~\ref{lem:L} gives
\[
2^m = \Omega\left( \frac{\log n}{\log r}
\right).
\]
This implies $m = \Omega(\log\log n)$ and concludes the proof of Theorem~\ref{thm:obstr1}.

\section{Obstructions for partial pairings}\label{sec:obstructionsgeneralk}

In this section we generalize the 
result of the previous section 
to partial pairings and show:

\newcommand{\kkkk}{k}

\begin{theorem}\label{thm:partial-pairing}
Let $n$ be an integer, $\kkkk\le n/2$,
and $\ket{\psi}$ be a $\kkkk$-pairable
state for $n$ parties where each party has $m$ qubits. Then
\[
\kkkk = O\left( n 2^{m} \frac{\log\log n}{\log n}\right)
\]
\end{theorem}

\begin{proof}
For technical reasons we assume that $n$ is divisible by 4. 
In the proof we also assume that $k\ge n/\log n$, since otherwise there is nothing to prove: the expression in parentheses on the right-hand side is always larger than $n/\log n$.
A $\kkkk$-pairing of $[n]$ is 
\[
\pi = \{\{a_{1},b_{1}\},\ldots,\{a_{\kkkk},b_{\kkkk}\}\} \;\;\; {\rm with} \;\;\; \cup \pi\subseteq [n], \;\;\;\;  |\cup \pi| =2\kkkk
\]
We denote the set of $\kkkk$-partial pairings 
on $[n]$ with $\Pi_{n,\kkkk}$. 
As in the previous section, we assume that
each party has $m$ qubits, and one of these~$m$ is designated 
as the output qubit, which will hold a qubit of an EPR-pair at the end of the protocol whenever $\pi$ involves the party in question. The goal is to be able to produce 
\[
|\pi\rangle = \ket{0^{n-2\kkkk}} \otimes \bigotimes_{i=1}^{\kkkk}\,  |\Phi^+\rangle_{a_{i},b_{i}} 
\]
from some fixed initial $nm$-qubit resource state $\ket{\psi}$, for all $\pi\in \Pi_{n,\kkkk}$.
We note that in the 
above tensor product the listing order of the qubits
depends on $\pi$,
and we list only the $n$ qubits designated to be output bits. (We list even those designated output qubits of parties that are not covered by the current partial matching $\pi$, since they will participate in the output for other $\pi$s.)
For the remaining $n(m-1)$ qubits, we assume w.l.o.g.\ that they 
end up in the $\ket{0}$-state, and hence are not entangled with the rest. To achieve this the parties can 
set these qubits to $\ket{0}$ by a local operation.

The proof is a slight variation of our proof for the complete-pairing case.
There $\dim \left({\cal M}_r \right)$ was calculated, and similarly to the previous section this dimension upper bounds the dimension of 
\[
{\cal L}_{\kkkk,r} =
{\rm span}\left(\; \ket{\pi}^{\otimes r} 
\mid\; 
\pi \in  \Pi_{n,\kkkk}
\right).
\]
The calculation is very similar to the case of complete pairings:
\begin{enumerate}
\item We will find $n^{\Theta(\kkkk)}$ different $\pi$s such that their
$r^{\rm th}$ tensor powers,
where $r=\Theta(\log n)$, are linearly independent.
(In the complete pairing case 
it was $n^{\Theta(n)}$ 
different $\pi$s.)
\item Setting $r = \Theta(\log n)$ is still
the only reasonable choice. Further, the approach breaks down 
at $m > \log \log n$, so 
we will be satisfied with 
investigating $m \le \log \log n$. With the above
parameters for $r$ and $m$ we have
$2^{m+1} + r -1 = \Theta(r)$, hence via Lemma~\ref{lem:dimMr} we have:
\[
\dim \left({\cal M}_r \right)\le{\Theta(r) \choose 2^{m} }^{n}\le 
r^{\Theta(2^{m}n)}.
\]
\item Similarly to our argument in the previous section, the 
dimension of ${\cal L}_{\kkkk,r}$
must lower bound the dimension of ${\cal M}_{C\log n}$, which is 
$r^{\Theta(2^{m}n)} = 2^{\Theta(n 2^{m}\log\log n) }$.
\item Combining 1 and 3, we get $\Theta(\kkkk \log n) \le \Theta(n\cdot 2^{m}\log\log n$), which implies Theorem~\ref{thm:partial-pairing}.
\end{enumerate}
Points 2-4 require no explanation as 
they just reiterate ideas of the 
previous section. However, we need to prove Point~1.

\medskip

First we prove the analogue of Lemma \ref{lem:in} for partial pairings.

\begin{lemma}\label{lem:in2}
Let $\pi$ and $\rho$ be two partial pairings with $\kkkk$ pairs. Then 
\[
\langle \pi |\rho\rangle = 2^{\mu-\kkkk}
\]
where $\mu$ is the number of cycles in $\pi\cup \rho$, as a graph on 
vertex set $[n]$.
\end{lemma}
\begin{proof}
We have:
\[
\langle \pi |\rho \rangle = \sum_{x\in \Lambda} \frac{1}{2^{\kkkk}}
\]
where $\Lambda\subseteq \{0,1\}^{n}$ is the set of binary strings, corresponding to the vertex labeling, $\lambda$, of the graph $\pi\cup \rho$
such that for every $\{a_{i},b_{i}\}\in \pi$ we have $\lambda(a_{i})= \lambda(b_{i})$, and
for every $\{a'_{i},b'_i\}\in \rho$ we have $\lambda(a'_{i})= \lambda(b'_{i})$ and {\it furthermore
every element of vertex set $[n]$ that is not covered by both a $\pi$-edge and a $\rho$-edge (that is, elements 
not in $(\cup \pi) \cap 
(\cup \rho)$) must get label 0.}\footnote{The formula for $\langle \pi |\rho \rangle$ comes from computing the inner product in the most straightforward way: we notice that both $\ket{\pi}$ and $\ket{\rho}$ have only two kinds of entries: 0 and
$1/\sqrt{2^{\kkkk}}$. Then  we just identify those entries where both $\ket{\pi}$ and 
$\ket{\rho}$ are non-zero and compute the number of such entries.}
Thus, only the cycles in $\pi\cup\rho$ can be labeled two ways, and no more than two ways, since the  edges of
$\pi$ and $\rho$ force the condition that all labels over the cycle must be either 0 or 1. (Paths cannot be labeled two ways as the label at their endpoint is fixed to 0.)
This calculation of $|\Lambda|$ gives the formula.
\end{proof}

\noindent We let $A = [n/2]$, $B = [n/2+1,n]$ and define
\[
 V \; = \; \left\{\, \pi \in  \Pi_{n,\kkkk} \mid  \mbox{$\pi$ is an $n$-qubit partial pairing with $a_{i}\in A$, $b_{i}\in B$ for $1\le i\le \kkkk$\,}\right\}
\]
Then $|V| = {n/2\choose \kkkk} (n/2)! / (n/2-\kkkk)!$. 
We need a lower bound on $|V|$. Using the  well-known bounds
\[
{n/2 \choose k} \ge \left(\frac{n}{2k}\right)^k \quad \mbox{and} \quad
k! \ge e^{-k} k^k
\]
one gets
\[
|V| = {n/2 \choose k}^2 k! \ge  \left(\frac{n}{2k}\right)^{2k} e^{-k} k^k
= \frac{n^{2k}}{k^k (4e)^k}
\ge  \frac{n^{2k}}{(n/2)^k (4e)^k}
= n^k(2e)^{-k},
\]
where the second inequality 
follows from $k\le n/2$. 
We again create a graph with:
\begin{eqnarray*}
V(G) & = & V \\
E(G) & = & \{ (\pi,\rho) \in V^{2} \mid \, \pi\neq \rho \; {\rm and} \; \pi\cup \rho \; \mbox{has at least $\kkkk/2$ cycles}\}
\end{eqnarray*}
Note that $G$ is again regular as in the previous section, since it is vertex-symmetric.
Like before, we want to lower bound $|V|/(D+1)$ where $D$ is the degree of a $\pi\in V$, which is 
then a lower bound on the size of a maximal independent set in~$G$. 
In fact, for an arbitrary fixed $\pi\in V$ we have:
\[
D + 1 = |\{\rho \in V \mid \pi \cup \rho \; \mbox{has at least $\kkkk/2$ cycles}\}|
\]
To upper bound the size of~$D$, w.l.o.g.\ let $\pi = \{\{1,1+n/2\},\ldots, \{\kkkk,\kkkk + n/2\}\}$: we match the first~$k$ vertices in the first half with the first~$k$ vertices in the second half. We can upper bound the number of neighbors of $\pi$
by enumerating them, each possibly multiple times. To define a somewhat elaborate enumeration, first notice that the nodes of 
any cycle in $\pi \cup \rho$ must be fully contained in the vertex set $\cup \pi = [1,\kkkk]\cup[1+n/2,\kkkk+n/2]$.
Assume $\rho$ is a neighbor of $\pi$ such that $\pi \cup \rho$ has $\mu$ cycles. Pick an (arbitrary) point in 
every cycle, such that the selected points belong to $A$. Let these points be $1\le p_{1} < \cdots < p_{\mu}\le \kkkk$,
and let
\[
P = \{p_{i}\}_{1\le i \le \mu}
\]
Let $2L_{i}$ be the length of the cycle that goes through $p_{i}$ (all cycles have even length, because 
the edges alternate between $\pi$ and $\rho$), and denote:
\begin{eqnarray*}
K_{\nu} & = & \sum_{i=1}^{\nu} L_{i}  \;\;\;\;\;  1\le \nu\le \mu \\
K & = & \{ K_{\nu} \}_{1\le \nu \le \mu}
\end{eqnarray*}
Since $K_{\mu}\le \kkkk$, we have $K\subseteq [\kkkk]$.
Finally, let us define
\[
R = \{a \in A\mid \; \mbox{$a$ is an $A$-endpoint of some edge in $\rho$}\}
\]
The triplet $(P,K,R)$ with one piece of additional information, which will be defined next, will determine~$\rho$.
The number of $(P,K,R)$ triplets is $2^{O(n)}$, since all three of 
$P,K,R$ can be given as subsets of sets of size at most $n$
(e.g., $K$ is a subset of $[n]$ due to $K_{\mu} \le n$).

Let us now understand the magnitude of the additional information\footnote{We could have chosen information-theoretic terminology
for our explanation, where we relate the $|V|/(D+1)$ ratio
to the mutual information between $\rho$ and $\pi$.
We have opted for an equivalent counting explanation,
but the reader has to bear in mind, that our underlying intuition is that every 
cycle that $\rho$ and $\pi$ jointly create, increases their mutual information by $\Omega(\log n)$ bits (essentially, the edge of $\rho$ that ``closes 
the cycle'' is cheap to communicate, given $\pi$).
The ratio $|V|/(D+1)$ is simply 
the exponential of this mutual information.} that together with $(P,K,R)$  
determines $\rho$. It will turn out that this information is
$c k \log n$ bits,
where very crucially, $c$ is less than~1 (in fact, $c$ will be essentially $1/2$).
This implies an upper bound on $D+1$, which in turn implies the following lower bound on the size of the largest independent set in $G$:
\begin{equation}\label{eq:lowerboundVD}
\frac{|V|}{D+1} \; \geq \; 
\underbrace{n^\kkkk(2e)^{-\kkkk}}_{\rm lower\; bound\; on\; |V|}
\; \cdot \underbrace{2^{-O(n)}}_{\rm \, due\; to \; {\it(P,K,R)}} \cdot \;
\underbrace{2^{-c\kkkk \log n}}_{\rm due\; to \;  additional \; information}
= n^{(1-c)\kkkk}(2e)^{-\kkkk}2^{-O(n)}.
\end{equation}

\medskip

\noindent\emph{Additional data that with $(P,K,R)$ uniquely determines $\rho$, given that $\rho$ is $\pi$'s neighbor in $G$:}

\medskip

We shall define an order $e_{1},\ldots,e_{\kkkk}$ of edges  of $\rho$. We will denote the $B$-endpoint of $e_{i}$ by~$q_{i}$. 

\medskip

Before telling the order, note that
the cycles in $\pi \cup \rho$ already have a
natural order, modulo the $(P,K,R)$
information, 
namely the $i^{\rm th}$ cycle is the 
cycle that contains $p_{i}$.
To ``extend'' this order to the edges we introduce:

\begin{enumerate}
\item Among all edges of $\rho$ those edges come earlier that are
edges of some $\pi\cup \rho$ cycle.
\item If two edges both participate in a $\pi\cup \rho$ cycle, but these two 
cycles are different, then the edge comes first
that belongs to an earlier cycle.
\end{enumerate}

We need to tell how to order edges within the same cycle. Also, we need to tell how to order edges that do not belong to any cycle.

\medskip

\noindent\emph{Ordering of edges of $\rho$ that belong to a given cycle.}\\ 
If the cycle is the $i^{\rm th}$
cycle, we simply walk through the  
the cycle and order the $\rho$-edges as we encounter them. 
The walk-through starts from $p_{i}$ with a $\rho$-edge (which determines that in which orientation 
we follow the cycle).
For instance, consider the first cycle.
For this example notice that if $e_{\ell}$ is a cycle-edge 
in $\pi \cup \rho$, then $\{q_{\ell}-n/2, q_{\ell}\}\in \pi$ belongs to the same cycle as $e_{\ell}$.
We get:

\begin{center}
\begin{tabular}{cl}
$e_{1}$ & is the edge of $\rho$ with $A$-endpoint $p_{1}$ \\
$e_{2}$ & is the edge of $\rho$ with $A$-endpoint $q_{1}-n/2$ \\
$e_{3}$ & is the edge of $\rho$ with $A$-endpoint $q_{2}-n/2$ \\
 & $\cdots$ \\
$e_{L_{1}}$ & is the edge of $\rho$ with $A$-endpoint $q_{L_{1}-1}-n/2$ \\
\end{tabular}
\end{center}

\medskip

\noindent\emph{Ordering of edges of $\rho$ that do not belong to any cycle:}

\medskip

These edges are simply ordered by the numerical value of their $A$-endpoints:
edges with a smaller $A$-endpoint come earlier.

\medskip

Let us now assume that Alice has to specify $\rho$ to Bob.
First Alice gives Bob the $(P,K,R)$ triplet.
Then she starts to tell Bob $q_{1}, q_{2},\ldots$. (Recall, $q_{i}$ is the $B$-endpoint of $e_{i}$.) Two remarkable observations lead to our conclusions.

\begin{enumerate}
    \item Even though Alice only tells the $B$-endpoints to Bob, Bob will (recursively) figure out the $A$-endpoints as well. This is in fact trivial. For instance, the $A$-endpoint of $e_{1}$ is $p_1$, which is known to Bob, since $(P,K,R)$ is given to him, etc.
    When the last cycle is exhausted, Bob
    knows this (this is after $K_{\mu}$ edges had been encountered---the total number of cycle-edges; $K_{\mu}$ in turn is given as the last element of the $K$-sequence), and from then on Bob relies on $R$
    to get the $A$-endpoints.
    \item When Alice arrives at an edge that closes a $\pi\cup\rho$ cycle,
    she does not need to send the $B$-endpoint of this edge! It is simply
    $p_i + n/2$, if the cycle was the $i^{\rm th}$ cycle. In other words, it is the other endpoint of the edge of $\pi$
    incident to $p_i$.
    Therefore, we just assume that Alice skips telling the $B$-endpoint of the last edge of every cycle.
    But how does Bob know that he has arrived at the last edge of the 
    current cycle? He knows this, because the cycle 
    lengths are encoded in~$K$: the length of the $i^{\rm th}$ 
    cycle is $K_{i}-K_{i-1}$ if $i\ge 2$ and $K_{1}$ for $i=1$.
\end{enumerate}

In summary, the information
Alice gives to Bob besides $(P,K,R)$ to identify~$\rho$,
is the $q_{1},q_{2},\ldots$ sequence, but crucially, completely leaving out from this sequence the $B$-endpoints of all the cycle-closing edges. We have $\mu$ cycle-closing edges.
Describing any $q_{\ell}$
takes $\log n$ bits, since $q_{\ell}\in [n]$.
Therefore:

\medskip

\noindent
\emph{The number of bits Alice needs to send Bob to fully describe a $\rho$ that creates $\mu$ cycles with $\pi$ is:}
\[
O(n) + (k-\mu)\log n
\]

\medskip

In conclusion, the number of different $\rho$s 
that form $\mu$ cycles with $\pi$ can be upper bounded by
$2^{O(n)}n^{\kkkk-\mu}$. To upper bound~$D$, recall that $\mu$ is allowed to vary from $\kkkk/2$ to~$\kkkk$, so
$D+1$ is upper bounded by $(\frac{\kkkk}{2} +1)2^ {O(n)}n^{\kkkk/2} + 1$. Thus, looking back at Eq.~(\ref{eq:lowerboundVD}) and using that $k\leq n$, there exists an independent set of size $n^{\kkkk/2}2^{-O(n)}$
in $G$.
Consequently, we can find a set of that many partial permutations such that
if $\pi\neq\rho$
belongs to this set,
then the inner product
of 
$\ket{\pi}^{\otimes r}$ and $\ket{\rho}^{\otimes r}$ is at most 
$2^{-\kkkk r/2}$. 
Setting $r$ to be 
$2\log n$ (generously, in fact: we care only about the $k\ge n\frac{\log \log n}{\log n}$ case, hence
the $2^{O(n)}$ factor becomes $n^{o(k)}$)
and applying Gershgorin's circle theorem in the same fashion
as in the proof of Lemma~\ref{lem:L} in the previous section, we 
prove Point~1 and conclude the proof of Theorem~\ref{thm:partial-pairing}.
\end{proof}

Theorem~\ref{thm:partial-pairing} has the following consequence for the case where each party is restricted to a constant number of qubits:

\begin{corollary}
Let $\ket{\psi}$ be a $\kkkk$-pairable
state for $n$ parties where each party has $m=O(1)$ qubits. Then
$\displaystyle\kkkk = O\left( n \frac{\log\log n}{\log n}\right)$.
\end{corollary}

As mentioned in the introduction, up to the power of the polylog this matches our expander-based construction of $k$-pairable states where $m=10$ and $k\geq n/\polylog(n)$ (Corollary~\ref{cor:mis10}).

\section{Conclusion and future work}

In this paper we initiated the study of $n$-party resource states from which LOCC protocols can create EPR-pairs between any $k$ disjoint pairs of parties. 
These EPR-pairs then enable quantum communication over a classical channel via teleportation.
Our focus was on the tradeoff between the number $k$ of to-be-created EPR-pairs (which we want to be large) and the number~$m$ of qubits per party (which we want to be small).

This work leaves open several questions for future work:
\begin{itemize}
\item Our constructions of $k$-pairable states may be far from optimal in various respects, and it would be interesting to improve them. 
For example, we already mentioned the follow-up papers~\cite{CMP:smallkpairable,CCMPST:vertexminor} which for the case $m=1$ reduced the exponential number of qubits $n(k)=2^{\Omega(k)}$ of our Reed-Muller-based construction of $k$-pairable states to a polynomial dependence on~$k$.
The case $m=2$ remains largely open for now, since two qubits per party is not enough to realize entanglement-swapping protocols on expander graphs, see Section~\ref{sec:constrmultiple}. 
All our constructions are based on stabilizer-type resource states.
Can one improve the tradeoff between $k$ and $m$ using more general resource states? Can one express the pairability parameter $k$ in terms of some previously studied entanglement measures?

\item Regarding lower bounds (obstructions), we showed in Section~\ref{sec:obstructionskisn/2} that a resource state for complete pairings ($n=k/2$) requires $m=\Omega(\log\log n)$ qubits per party. Can we improve this lower bound to $m=\Omega(\log n)$ qubits, matching the upper bound we obtained from expander graphs at the end of Section~\ref{sec:constrmultiple}?
Our lower bounds are actually for a stronger model, applying to LOCC protocols that produce the desired state with positive probability; there may be better upper bounds in this setting, and/or stronger lower bounds for LOCC protocols that are required to succeed with probability~1.
\item How well do our resource states behave under noise? Contreras-Tejada, Palazuelos, and de Vicente~\cite{CPV:noisynetworks} already proved some negative results here for the type of constructions we gave in Section~\ref{sec:constrmultiple} (with EPR-pairs on the edges of an $n$-vertex graph), showing that genuine multipartite entanglement only survives constant amounts of noise per edge if the graph has a lot of connectivity. 
\item We can ask a very similar \emph{classical} question, where the classical analogue of an EPR-pair is a uniform bit shared between two parties and unknown to all others. 
Such shared secret bits can then be used for secure communication over public classical channels (via the one-time pad), similarly to how shared EPR-pairs can be used for secure quantum communication over public classical channels (via teleportation).
We believe our techniques can be modified to obtain non-trivial results about the question: what classically correlated $n$-party resource states are necessary and sufficient for LOCC protocols (with public classical communication) to generate such secret shared bits between any $k$ disjoint pairs of parties? 
One difference is that the straightforward classical analogue of the GHZ-state (a uniformly random bit known to all $n$ parties) is not 1-pairable in this classical sense.
\end{itemize}

\paragraph*{Acknowledgements.}
We thank Carlos Palazuelos for a pointer to~\cite{CPV:noisynetworks}, Jorge Miguel-Ramiro for a pointer to~\cite{MPD:networks}, and Mehdi Mhalla for pointers to~\cite{CMP:smallkpairable,CCMPST:vertexminor}.

\bibliographystyle{plainnat}
\bibliography{ref}

\begin{thebibliography}{36}
\providecommand{\natexlab}[1]{#1}
\providecommand{\url}[1]{\texttt{#1}}
\expandafter\ifx\csname urlstyle\endcsname\relax
  \providecommand{\doi}[1]{doi: #1}\else
  \providecommand{\doi}{doi: \begingroup \urlstyle{rm}\Url}\fi

\bibitem[Aaronson and Gottesman(2004)]{aaronson2004improved}
Scott Aaronson and Daniel Gottesman.
\newblock Improved simulation of stabilizer circuits.
\newblock \emph{Physical Review A}, 70\penalty0 (5):\penalty0 052328, 2004.
\newblock \doi{10.1103/physreva.70.052328}.

\bibitem[Adcock et~al.(2020)Adcock, Morley-Short, Dahlberg, and
  Silverstone]{adcock2020mapping}
Jeremy~C. Adcock, Sam Morley-Short, Axel Dahlberg, and Joshua~W. Silverstone.
\newblock Mapping graph state orbits under local complementation.
\newblock \emph{Quantum}, 4:\penalty0 305, 2020.
\newblock \doi{10.22331/q-2020-08-07-305}.

\bibitem[Assmus~Jr(1992)]{assmus1992reed}
E.~F. Assmus~Jr.
\newblock On the {R}eed-{M}uller codes.
\newblock \emph{Discrete mathematics}, 106:\penalty0 25--33, 1992.
\newblock \doi{10.1016/0012-365X(92)90526-L}.

\bibitem[Bennett et~al.(1993)Bennett, Brassard, Cr{\'e}peau, Jozsa, Peres, and
  Wootters]{teleporting}
Charles Bennett, Gilles Brassard, Claude Cr{\'e}peau, Richard Jozsa, Asher
  Peres, and William Wootters.
\newblock Teleporting an unknown quantum state via dual classical and
  {Einstein-Podolsky-Rosen} channels.
\newblock \emph{Physical Review Letters}, 70:\penalty0 1895--1899, 1993.
\newblock \doi{10.1103/PhysRevLett.70.1895}.

\bibitem[Bollob{\'a}s(1988)]{bollobas1988isoperimetric}
B{\'e}la Bollob{\'a}s.
\newblock The isoperimetric number of random regular graphs.
\newblock \emph{European Journal of combinatorics}, 9\penalty0 (3):\penalty0
  241--244, 1988.
\newblock \doi{10.1016/S0195-6698(88)80014-3}.

\bibitem[Bombin and Martin-Delgado(2006)]{bombin2006topological}
Hector Bombin and Miguel~Angel Martin-Delgado.
\newblock Topological quantum distillation.
\newblock \emph{Physical review letters}, 97\penalty0 (18):\penalty0 180501,
  2006.
\newblock \doi{10.1103/PhysRevLett.97.180501}.

\bibitem[Broder et~al.(1994)Broder, Frieze, and Upfal]{broder1994existence}
Andrei~Z. Broder, Alan~M. Frieze, and Eli Upfal.
\newblock Existence and construction of edge-disjoint paths on expander graphs.
\newblock \emph{SIAM Journal on Computing}, 23\penalty0 (5):\penalty0 976--989,
  1994.
\newblock \doi{10.1145/129712.129727}.

\bibitem[Calderbank et~al.(1998)Calderbank, Rains, Shor, and
  Sloane]{calderbank1998quantum}
A.~R. Calderbank, E.~M. Rains, P.~W. Shor, and N.~J.~A. Sloane.
\newblock Quantum error correction via codes over {GF}(4).
\newblock \emph{IEEE Transactions on Information Theory}, 44\penalty0
  (4):\penalty0 1369--1387, 1998.
\newblock \doi{10.48550/arXiv.quant-ph/9608006}.

\bibitem[Calderbank and Shor(1996)]{calderbank1996good}
A.~Robert Calderbank and Peter~W. Shor.
\newblock Good quantum error-correcting codes exist.
\newblock \emph{Physical Review A}, 54\penalty0 (2):\penalty0 1098, 1996.
\newblock \doi{10.1103/PhysRevA.54.1098}.

\bibitem[Caut{\'e}s et~al.(2024)Caut{\'e}s, Claudet, Mhalla, Perdrix, Savin,
  and Thomass{\'e}]{CCMPST:vertexminor}
Maxime Caut{\'e}s, Nathan Claudet, Mehdi Mhalla, Simon Perdrix, Valentin Savin,
  and St{\'e}phan Thomass{\'e}.
\newblock Vertex-minor universal graphs for generating entangled quantum
  subsystems.
\newblock arXiv:2402.06260, 2024.
\newblock URL \url{https://doi.org/10.48550/arXiv.2402.06260}.

\bibitem[Claudet et~al.(2023)Claudet, Mhalla, and Perdrix]{CMP:smallkpairable}
Nathan Claudet, Mehdi Mhalla, and Simon Perdrix.
\newblock Small $k$-pairable states.
\newblock arXiv:2309.09956, 2023.
\newblock URL \url{https://doi.org/10.48550/arXiv.2309.09956}.

\bibitem[Cleve and Buhrman(1997)]{CleveBuhrman1997}
R.~Cleve and H.~Buhrman.
\newblock Substituting quantum entanglement for communication.
\newblock \emph{Physical Review A}, 56\penalty0 (2):\penalty0 1201--1204, 1997.
\newblock \doi{10.1103/PhysRevA.56.1201}.

\bibitem[{Contreras-Tejada} et~al.(2022){Contreras-Tejada}, Palazuelos, and
  de~{Vicente}]{CPV:noisynetworks}
Patricia {Contreras-Tejada}, Carlos Palazuelos, and Julio~I. de~{Vicente}.
\newblock Asymptotic survival of genuine multipartite entanglement in noisy
  quantum networks depends on the topology.
\newblock \emph{Physical Review Letters}, 128:\penalty0 220501, 2022.
\newblock \doi{10.1103/PhysRevLett.128.220501}.
\newblock arXiv:2106.04634.

\bibitem[Dahlberg et~al.(2020{\natexlab{a}})Dahlberg, Helsen, and
  Wehner]{Dahlberg_2020}
Axel Dahlberg, Jonas Helsen, and Stephanie Wehner.
\newblock How to transform graph states using single-qubit operations:
  computational complexity and algorithms.
\newblock \emph{Quantum Science and Technology}, 5\penalty0 (4):\penalty0
  045016, sep 2020{\natexlab{a}}.
\newblock \doi{10.1088/2058-9565/aba763}.

\bibitem[Dahlberg et~al.(2020{\natexlab{b}})Dahlberg, Helsen, and
  Wehner]{dahlberg2020transforming}
Axel Dahlberg, Jonas Helsen, and Stephanie Wehner.
\newblock Transforming graph states to {B}ell-pairs is {NP}-complete.
\newblock \emph{Quantum}, 4:\penalty0 348, 2020{\natexlab{b}}.
\newblock \doi{10.22331/q-2020-10-22-348}.

\bibitem[Du et~al.(2017)Du, Shang, and Liu]{DSL:multipoint}
Gang Du, Tao Shang, and Jian-Wei Liu.
\newblock Quantum coordinated multi-point communication based on entanglement
  swapping.
\newblock \emph{Quantum Information Processing}, 2017.
\newblock \doi{10.1007/s11128-017-1558-2}.

\bibitem[Fattal et~al.(2004)Fattal, Cubitt, Yamamoto, Bravyi, and
  Chuang]{Fattal2004EntanglementIT}
David Fattal, Toby~S. Cubitt, Yoshihisa Yamamoto, Sergey Bravyi, and Isaac~L.
  Chuang.
\newblock Entanglement in the stabilizer formalism.
\newblock \emph{arXiv: Quantum Physics}, 2004.
\newblock \doi{10.48550/arXiv.quant-ph/0406168}.

\bibitem[Fischer and Townsley(2021)]{fischer&townsley21}
Alex Fischer and Don Townsley.
\newblock Distributing graph states across quantum networks.
\newblock In \emph{Proceedings of IEEE International Conference on Quantum
  Computing and Engineering (QCE)}, 2021.
\newblock \doi{10.1109/QCE52317.2021.00049}.

\bibitem[Fitzi et~al.(2001)Fitzi, Gisin, and Maurer]{2001_Fitzi}
Matthias Fitzi, Nicolas Gisin, and Ueli Maurer.
\newblock Quantum solution to the {B}yzantine agreement problem.
\newblock \emph{Physical Review Letters}, 87\penalty0 (21), nov 2001.
\newblock \doi{10.1103/physrevlett.87.217901}.

\bibitem[Gottesman(1998{\natexlab{a}})]{gottesman1998heisenberg}
Daniel Gottesman.
\newblock The {H}eisenberg representation of quantum computers.
\newblock \emph{arXiv preprint quant-ph/9807006}, 1998{\natexlab{a}}.
\newblock \doi{10.48550/arXiv.quant-ph/9807006}.

\bibitem[Gottesman(1998{\natexlab{b}})]{gottesman1998theory}
Daniel Gottesman.
\newblock Theory of fault-tolerant quantum computation.
\newblock \emph{Physical Review A}, 57\penalty0 (1):\penalty0 127,
  1998{\natexlab{b}}.
\newblock \doi{10.1103/PhysRevA.57.127}.

\bibitem[Gy{\H{o}}ri et~al.(2017)Gy{\H{o}}ri, Mezei, and
  M{\'e}sz{\'a}ros]{gyHori2017note}
Ervin Gy{\H{o}}ri, Tam{\'a}s~R{\'o}bert Mezei, and G{\'a}bor M{\'e}sz{\'a}ros.
\newblock Note on terminal-pairability in complete grid graphs.
\newblock \emph{Discrete Mathematics}, 340\penalty0 (5):\penalty0 988--990,
  2017.
\newblock \doi{10.1016/j.disc.2017.01.014}.

\bibitem[Hahn et~al.(2019)Hahn, Pappa, and Eisert]{hahn2019quantum}
Frederik Hahn, Anna Pappa, and Jens Eisert.
\newblock Quantum network routing and local complementation.
\newblock \emph{npj Quantum Information}, 5\penalty0 (1):\penalty0 1--7, 2019.
\newblock \doi{10.1038/s41534-019-0191-6}.

\bibitem[Hillery et~al.(1999)Hillery, Bu{\v{z}}ek, and
  Berthiaume]{PhysRevA.59.1829}
Mark Hillery, Vladimir Bu{\v{z}}ek, and Andr{\'e} Berthiaume.
\newblock Quantum secret sharing.
\newblock \emph{Physical Review A}, 59:\penalty0 1829--1834, Mar 1999.
\newblock \doi{10.1103/PhysRevA.59.1829}.
\newblock URL \url{https://link.aps.org/doi/10.1103/PhysRevA.59.1829}.

\bibitem[Illiano et~al.(2022)Illiano, Viscardi, Koudia, Caleffi, and
  Cacciapuoti]{10008516}
Jessica Illiano, Michele Viscardi, Seid Koudia, Marcello Caleffi, and
  Angela~Sara Cacciapuoti.
\newblock Quantum internet: from medium access control to entanglement access
  control.
\newblock In \emph{2022 IEEE Globecom Workshops (GC Wkshps)}, pages 1329--1334,
  2022.
\newblock \doi{10.1109/GCWkshps56602.2022.10008516}.

\bibitem[MacWilliams and Sloane(1977)]{macwilliams1977theory}
Florence~Jessie MacWilliams and Neil James~Alexander Sloane.
\newblock \emph{The theory of error correcting codes}, volume~16.
\newblock Elsevier, 1977.
\newblock \doi{10.1137/1022103}.

\bibitem[Meignant et~al.(2019)Meignant, Markham, and
  Grosshans]{MMG:distributing}
Cl\'{e}ment Meignant, Damian Markham, and Fr\'{e}d\'{e}ric Grosshans.
\newblock Distributing graph states over arbitrary quantum networks.
\newblock \emph{Physical Review A}, 100\penalty0 (052333), 2019.
\newblock \doi{10.1103/PhysRevA.100.052333}.
\newblock arXiv:1811.05445.

\bibitem[Meter(2014)]{VanMeter2014}
Rodney~Van Meter.
\newblock \emph{Quantum Networking}.
\newblock John Wiley \& Sons, Ltd, 2014.
\newblock ISBN 9781848215375.
\newblock \doi{10.1002/9781118648919}.

\bibitem[{Miguel-Ramiro} et~al.(2023){Miguel-Ramiro}, Pirker, and
  D{\"u}r]{MPD:networks}
Jorge {Miguel-Ramiro}, Alexander Pirker, and Wolfgang D{\"u}r.
\newblock Optimized quantum networks.
\newblock \emph{Quantum}, 7:\penalty0 919, 2023.
\newblock \doi{10.22331/q-2023-02-09-919}.
\newblock arXiv:2107.10275.

\bibitem[Nielsen and Chuang(2002)]{nielsen2002quantum}
Michael~A. Nielsen and Isaac Chuang.
\newblock \emph{Quantum computation and quantum information}.
\newblock Cambridge University Press, 2002.
\newblock \doi{10.1119/1.1463744}.

\bibitem[Pant et~al.(2019)Pant, Krovi, Towsley, Tassiulas, Jiang, Basu,
  Englund, and Guha]{pant2019routing}
Mihir Pant, Hari Krovi, Don Towsley, Leandros Tassiulas, Liang Jiang, Prithwish
  Basu, Dirk Englund, and Saikat Guha.
\newblock Routing entanglement in the quantum internet.
\newblock \emph{npj Quantum Information}, 5\penalty0 (1):\penalty0 1--9, 2019.
\newblock \doi{10.1038/s41534-019-0139-x}.

\bibitem[Raussendorf and Briegel(2001)]{PhysRevLett.86.5188}
Robert Raussendorf and Hans~J. Briegel.
\newblock A one-way quantum computer.
\newblock \emph{Physical Review Letters}, 86:\penalty0 5188--5191, May 2001.
\newblock \doi{10.1103/PhysRevLett.86.5188}.

\bibitem[Raussendorf et~al.(2003)Raussendorf, Browne, and
  Briegel]{raussendorf2003measurement}
Robert Raussendorf, Daniel~E. Browne, and Hans~J. Briegel.
\newblock Measurement-based quantum computation on cluster states.
\newblock \emph{Physical Review A}, 68\penalty0 (2):\penalty0 022312, 2003.
\newblock \doi{10.1103/PhysRevA.68.022312}.

\bibitem[Schoute et~al.(2016)Schoute, Man{\v{c}}inska, Islam, Kerenidis, and
  Wehner]{schouteetal:shortcuts}
Eddie Schoute, Laura Man{\v{c}}inska, Tanvirul Islam, Iordanis Kerenidis, and
  Stephanie Wehner.
\newblock Shortcuts to quantum network routing.
\newblock arXiv:1610.05238, 2016.
\newblock URL \url{https://doi.org/10.48550/arXiv.1610.05238}.

\bibitem[Steane(1996)]{steane1996multiple}
Andrew Steane.
\newblock Multiple-particle interference and quantum error correction.
\newblock \emph{Proceedings of the Royal Society of London. Series A:
  Mathematical, Physical and Engineering Sciences}, 452\penalty0
  (1954):\penalty0 2551--2577, 1996.
\newblock \doi{10.1098/rspa.1996.0136}.

\bibitem[Wehner et~al.(2018)Wehner, Elkouss, and Hanson]{WEH:qinternet}
Stephanie Wehner, David Elkouss, and Ronald Hanson.
\newblock Quantum internet: A vision for the road ahead.
\newblock \emph{Science}, 362\penalty0 (6412), 2018.
\newblock \doi{10.1126/science.aam9288}.

\end{thebibliography}

\appendix

\section{Constructing an $f$ that satisfies condition CSS1, when all $a_{i}$, $b_{i}$
are independent}
\label{app:example}

In this appendix we give a proof for a simpler but instructive special case of how we satisfy CSS1 in Section~\ref{subs:k-pairable}.
The additional assumption of our special case is to suppose that $a_1,b_1,\ldots,a_k,b_k\in \FF_2^m$ are linearly independent.

Let $\calL$ be the linear subspace of $\FF_2^m$ spanned by $c_1,\ldots,c_k$
and $a_1,b_1,\ldots,a_k,b_k$ (so in particular, $m=3k$). Let $x_{1},\ldots,x_{m}$ be (mod~2) variables.
Then the above $3k$ vectors define $3k$ linear functions
over $\calL$:
every vector $x\in \calL$ can be uniquely written as
\[
x=\sum_{i=1}^k \alpha_i a_i + \beta_i b_i + \gamma_i c_i
\]
for some binary coefficients $\alpha_i,\beta_i,\gamma_i$ that are 
functions of $x = (x_{1},\ldots,x_{m})$.
As such, all these functions are linear, since 
when writing down $x+x'$ as above, we add the corresponding coefficients.
In the sequel we shall create higher-degree polynomials 
over $x_{1},\ldots,x_{m}$
from these linear functions (e.g., $\alpha_1(x)\beta_1(x)$ is 
a quadratic function).

Let $S\equiv \calS_1\cup \calS_2 \cup \ldots \cup \calS_k$.
Note that $S\subseteq \calL$.
Condition CSS1 (with $i=1$) asks for a degree-$(k-1)$ polynomial $f$
such that 
\be
\label{app_CSS1}
f(a_1)=f(b_1)=1 \quad \mbox{and} \quad 
f(x)=0 \quad \mbox{for all $x\in S\setminus \{a_1,b_1\}$}.
\ee
Let us show that Eq.~(\ref{app_CSS1}) is satisfied if we choose $f$ as
\be
\label{app_eq1}
f(x)=(\alpha_1(x)+\beta_1(x))g(x) {\pmod 2},
\ee
where 
\be
\label{app_g(x)}
g(x) = \sum_{M\subsetneq [k]} \; \prod_{j\in M} \gamma_j(x) {\pmod 2}.
\ee
The polynomial $f$ defined in Eq.~(\ref{app_eq1}) has degree $k$, because $\alpha_1(x)+\beta_1(x)$ has degree~1 and $g(x)$ has degree~$k-1$. However, we will see that the restriction of $f$ onto $S$ coincides with a degree-$(k-1)$ polynomial.

First, we show that
\[
g(x)=\left\{ \ba{rl}
1 &\mbox{if }  
\gamma_1(x)=\gamma_2(x)=\cdots=\gamma_k(x),\\
0 & \mbox{otherwise}\\
\ea
\right.
\]
Indeed, extending the sum over $M$ in Eq.~(\ref{app_g(x)}) to all subsets
$M\subseteq [k]$ would give a function $\prod_{j=1}^k(1+\gamma_j(x))$
which is zero mod~2 unless $\gamma_j(x)=0$ for all $j$. 
The missing monomial $\prod_{j=1}^k \gamma_j(x)$ associated with the 
subset $M=[k]$
is zero unless $\gamma_j(x)=1$ for all $j$.
 
By definition of $\calS_j$,
any vector $x\in \calS_j$ can be written as a sum
of an odd number of vectors
from the set $\{a_j,b_j,c_1,\ldots,c_k\}\setminus \{c_j\}$.
In particular, $\gamma_j(x)=0$ for any $x\in \calS_j$.
Thus the restriction of $g(x)$ onto $S$
is zero unless $\gamma_j(x)=0$ for all $j$.
In the latter case one has $x=a_j$ or $x=b_j$ for some~$j\in[k]$.
If $j=1$ then $f(x)=1$ since $\alpha_1(x)+\beta_1(x)=1$.
If $j\ge 2$ then $f(x)=0$ since $\alpha_1(x)=\beta_1(x)=0$.
This proves Eq.~(\ref{app_CSS1}).

Next we claim that degree-$k$ monomials in $f(x)$  can be replaced by monomials
of degree at most $k-1$ 
without changing the restriction of $f$ onto $S$. Indeed, 
the sum of all degree-$k$ monomials in $f(x)$ can be written as 
\be
\label{f'(x)}
f'(x) = (\alpha_1(x)+\beta_1(x))\sum_{i=1}^k\;  \prod_{j \in [k]\setminus \{i\}} \gamma_j(x) {\pmod 2}.
\ee
Supose $x\in S$. We claim that $f'(x)=1$ iff  $x=a_1+c_2+\cdots+c_k$ or $x=b_1+c_2+\cdots+c_k$.
Indeed, if $x\in \calS_j$ for $j\ge 2$ then $f'(x)=0$ since $\alpha_1(x)=\beta_1(x)=0$.
Suppose $x\in \calS_1$. Then $\gamma_1(x)=0$ and thus 
$f'(x)= (\alpha_1(x)+\beta_1(x))\gamma_2(x)\cdots \gamma_k(x){\pmod 2}$.
By definition of $\calS_1$, one can write 
$x$ as a sum of an odd number of vectors
from the set $\{a_1,b_1,c_2,\ldots,c_k\}$. Hence $f'(x)=1$ iff
$x=a_1+c_2+\cdots+c_k$ or $x=b_1+c_2+\cdots+c_k$.

If $k$ is even then none of the vectors $a_1+c_2+\cdots+c_k$
and $b_1+c_2+\cdots+c_k$ belongs to $\calS_1$. Thus $f'(x)=0$ for all $x\in S$.
If $k$ is odd, the same arguments as above show that 
\be
f'(x)=\prod_{j=2}^k \gamma_j(x) \quad \mbox{for any} \quad x\in S.
\ee
Thus we can replace $f'(x)$ by a monomial of degree either zero (if $k$ is even)
or degree $k-1$ (if $k$ is odd) without changing the restriction of $f$ onto $S$.
This proves condition CSS1.

\medskip

\noindent{\bf A numerical example:}
Let $k=4$, $m=12$, and let

\begin{center}
    \begin{tabular}{lll}
      $a_{1}\;=$   &  000010000000 \\
      $b_{1}\;=$   &  000001000000 \\
      $a_{2}\;=$   &  000000100000 \\
      $b_{2}\;=$   &  000000010000 \\
      $a_{3}\;=$   &  000000001000 \\
      $b_{3}\;=$   &  000000000100 \\
      $a_{4}\;=$   &  000000000010 \\
      $b_{4}\;=$   &  000000000001 \\
    \end{tabular}
\end{center}

We also set $c_{i}=x_{i}$ for $1\le i\le 4$, so $c_{1}= 100000000000$, etc.
To address condition CSS1 for $i=1$
we want to write down a degree-3 polynomial $f(x_{1},\ldots,x_{12})$ over $\FF_{2}$ 
that on the set $EZ = \calS_1\cup \ldots \cup \calS_k$
takes value 1 on $a_1$ and $b_1$, and 0 on the rest of $EZ$.
(By symmetry then we can write down similar functions for $i=2,3,4$.)
Let $w_{4}(x)$ be the Hamming weight of the first 4 bits of $x$ and
let $w_{-8}(x)$ be the 
Hamming weight of 
the last 8 bits of $x$.
Notice that for all $x\in EZ$ we have $w_{4}(x)\le 3$ and $w_{-8}(x)\le 2$,
and even within this restriction some weight combinations
$(w_{4}, w_{-8}(x))$ may never arise for any element of $EZ$.
Such weight combinations we will call ``impossible pairs.'' 
For each weight combination we have calculated the number of elements 
of EZ that have that weight combination, see the table below. We put an asterisk (rather than 0) in the entries that represent impossible weight combinations.

\begin{center}
\begin{tabular}{|l||*{4}{c|}}\hline
\backslashbox{$w_{-8}$}{$w_{4}$}
&\makebox[3em]{0}&\makebox[3em]{1}&\makebox[3em]{2}
&\makebox[3em]{3} \\\hline\hline
0  & $\ast$ & 4 & $\ast$ & 4\\\hline
1 & 8 & $\ast$ & 24 & $\ast$ \\\hline
2 & $\ast$ & 12 & $\ast$ & 4\\\hline
\end{tabular}
\end{center}

The entries are not hard to calculate, and to give a typical example 
we calculate the 
$(1,2)$-entry: There are 8 weight-1 strings of the last 8 bits, and 6 weight-2 strings 
of the first 4 bits.
However, not all elements of the set $\{a_{1},b_{1},a_{2},\ldots,b_{4}\}$
can be added to some element of the set $\{c_{1}+c_{2},\ldots, c_{3}+c_{4}\}$.
For instance $c_{1}+c_{2}+a_{1}$ does not occur in $EZ$, as $\calS_1$ does not contain $c_1$.
It is easy to see that these bad combinations are
half of all possible $8\cdot 6=48$ combinations, hence we obtain 24 as the $(1,2)$-entry of the table.

Observe now that $|EZ|=56> 16$, so 
Lemma~\ref{lemma:reg} cannot be applied directly to get a degree-3 polynomial.
Let us now construct a degree-3 polynomial $f:\FF_{2}^{12}\rightarrow \FF_{2}$
that \emph{restricted to $EZ$} satisfies condition
CSS1 for $i=1$. First note that $f$ satisfies CSS1 for $i=1$ if
\[
\forall x\in EZ:\;\; f(x) = 1 \; \Longleftrightarrow \;
w_{4}(x)=0\; \wedge \; x_{5}+x_{6}=1.
\]
Consider now the polynomial 
\[
g(x) = 1 + \sum_{i=1}^{4}x_{i} + \sum_{1\le i < j \le 4} x_{i}x_{j}
\]
One can easily check that
\begin{center}
    \begin{tabular}{lccccc}
      When $w_{4}(x)$ & $=$   &  0 & 1 & 2 & 3 \\
      then $g(x)$ mod~2 & $=$ & 1 & 0 & 0 & 1
    \end{tabular}
\end{center}

\noindent and that the polynomial $f(x) = \left(x_{5}+x_{6}\right) g(x)$ then takes values on $EZ$
exactly as needed. (For instance, when $w_{4}=3$ {\it and $x\in EZ$} then 
$x_{5}+x_{6}$ will always give zero, etc.)

\section{Induction step of Section~\ref{subs:k-pairable}: modifications for odd $k$}
\label{app:odd_k}

Here we extend the proof
of $k$-pairability 
given in Section~\ref{subs:k-pairable}
to odd values of $k$.
Suppose $\ell\ge 0$ is an integer.
We say that a family of polynomials
$f_T\, : \, \FF_2^m\to \FF_2$ labeled by subsets
$T\subsetneq [k]$ with $|T|\le \ell$ 
is \emph{valid}  if it satisfies the following
conditions.
\begin{enumerate}
\item[\bf I1:] $f_T$ depends only on the variables $x_{k+1},\ldots,x_m$
\item[\bf I2:]  $f_T$ has degree $k-1-|T|$
\item[\bf I3:]  $f_\emptyset(a_1)=f_\emptyset(b_1)=1$ and $f_\emptyset(a_i)=f_\emptyset(b_i)=0$ for $2\le i\le k$
\item[\bf I4:]  $f_T(x)=\sum_{U\subsetneq T} f_U(x)$ for any non-empty set $T\subsetneq [k]$ and any $x\in e(T)$
\item[\bf I5:]  $f_T\equiv 0$ if $|T|$ is odd and $|T|\le k-3$
\item[\bf I6:]  $f_T(0^m)=f_T(a_i+b_i)=0$ if $i\notin T$ and $|T|\le k-4$
\item[\bf I7:] $f_T(a_i)=f_T(b_i)$ if $i\notin T$
and $|T|\le k-3$
\end{enumerate}
These conditions are identical to the ones
given in Section~\ref{subs:k-pairable}, except for
condition (I7) which is now imposed 
only for $|T|\le k-3$.
Below we assume that $k$ is odd. 
As before,
we shall use induction on $\ell$ to prove that a valid family of polynomials 
exists for all $\ell\le k-1$. 

The base of induction is $\ell=0$.
Then a valid family is a single polynomial
$f_\emptyset$.  
The construction of $f_\emptyset$ is identical
to the one given in Section~\ref{subs:k-pairable}.

We shall now prove the induction step.
Suppose we have already constructed a valid family of polynomials $f_T$
with $|T|\le \ell-1$. 
Consider a subset $T\subsetneq [k]$ with $|T|=\ell$ such that $1\le \ell\le k-2$. The case $\ell=k-1$
will be considered afterwards.

Suppose $|T|$ is odd. 
Condition (I2) demands that $f_T$ has degree $d_T=k-1-|T|$. Condition (I3) can be skipped since $T\ne \emptyset$. Consider two cases.

\noindent
\emph{Case 1:} $|T|\le k-3$.  Then (I5) demands
$f_T\equiv 0$. This automatically satisfies
(I6,I7). It remains to check (I4) which is equivalent
to 
\be
\label{extra_condition_5}
\sum_{U\subsetneq T} f_U(x)=0
\ee
for $x\in e(T)$. Note that $e(T)$ consists of points
$0^m$ and $a_i+b_i$ with $i\notin T$.
All terms $f_U(x)$ with odd $|U|$ vanish due
to (I5). All terms $f_U(x)$ with even $|U|$
obey (I6). Since $|U|\le |T|-1\le k-4$, we have
$f_U(0^m)=0$ and $f_U(a_i+b_i)=0$ due to (I6).
Thus all terms $f_U(x)=0$ vanish, that is,
Eq.~(\ref{extra_condition_5}) is satisfied.

\noindent
\emph{Case 2:} $|T|\ge k-2$. Then $|T|=k-2$ since both
$k$ and $|T|$ are odd. We can skip  (I5,I6,I7).
Thus we just need to satisfy (I4). 
It fixes the 
values of $f_T$ at $s=|e(T)|=k-|T|+1=3$ points,
namely, $0^m$ and $a_i+b_i$ with $i\notin T$.
Since $s=3<2^{d_T+1}=2^{k-|T|}=4$,
the desired polynomial $f_T$ exists by part~(1)
of Lemma~\ref{lemma:reg}.

Suppose $|T|$ is even. Condition (I2)
demands that $f_T$ has degree $d_T=k-1-|T|$.
Conditions (I3),(I5) can be skipped
since $|T|$ is even and $T\ne \emptyset$.
We claim that (I7) follows from (I4).
Indeed, we have $x\in e(T)$ iff
$x=a_i$ or $x=b_i$ with $i\notin T$.
We have $f_T(x)=\sum_{U\subsetneq T} f_U(x)$
due to (I4). If $|U|$ is odd then
$f_U\equiv 0$ due to (I5) since 
$|U|\le |T|-1=\ell-1\le k-3$. If $|U|$ is even then $|U|\le |T|-2\le k-4$.
Thus $f_U(a_i)=f_U(b_i)$ since (I7) holds for $f_U$.
This implies $f_T(a_i)=f_T(b_i)$ for any $i\notin T$, as claimed in (I7).
It remains to check (I4,I6).
We fix the value of $f_T$ at $s=|e(T)|=2(k-|T|)$ points in (I4) 
if $|T|\ge k-3$. We fix the value of $f_T$ at $s=|e(T)|+k-|T|+1 = 3(k-|T|)+1$
points in (I4,I6) if $|T|\le k-4$.
Consider first the case $|T|\ge k-3$.
Then $|T|=k-3$ since $k$ is odd,
$|T|$ is even, and $|T|=\ell\le k-2$.
Thus $s=2(k-|T|)=6<2^{d_T+1}=2^{k-|T|}=8$.
Part~(1) of Lemma~\ref{lemma:reg}
implies that the desired polynomial $f_T$ exists.
Next consider the case $|T|\le k-4$.
Then $s=3(k-|T|)+1<2^{d_T+1}=2^{k-|T|}$.
Part~(1) of Lemma~\ref{lemma:reg}
implies the desired polynomial $f_T$ exists.

It remains to prove the induction step for $\ell=k-1$.
Suppose we have already constructed a valid family of polynomials $f_T$
with $|T|\le k-2$. 
Consider a subset $T\subsetneq [k]$ with $|T|=k-1$.
Since we assumed that $k$ is odd, $|T|$ is even.
Condition (I2) demands that $f_T$
has  degree $k-1-|T|=0$, that is, $f_T$ is
a constant function.
We can skip conditions (I3,I5,I6,I7)
since none of them applies if $|T|=k-1$.
It remains to check (I4).
Note that $e(T)=\{a_i,b_i\}$ for some $i\in [k]$ such that $T=[k]\setminus \{i\}$.
Condition (I4) fixes the value of $f_T$ at
$a_i$ and $b_i$.
Since we want $f_T$ to be a constant function, it suffices to check that the desired
values $f_T(a_i)$
and $f_T(b_i)$ are the same.
Substituting the desired values from (I4),
we have to check that 
\be
\label{extra_condition_6}
\sum_{U\subsetneq T} f_U(a_i) + f_U(b_i)=0.
\ee
All terms $f_U(a_i)+f_U(b_i)$ with $|U|\le k-3$
vanish since $f_U$ obeys (I7).
Thus we can restrict the sum Eq.~(\ref{extra_condition_6}) to terms
with $|U|=k-2$.
However $f_U$ is a degree-1 polynomial if $|U|=k-2$ due to (I2). 
Thus $f_U(a_i)+f_U(b_i) = f_U(0^m)+f_U(a_i+b_i)$
and Eq.~(\ref{extra_condition_6}) is equivalent to
\be
\label{extra_condition_7}
\sum_{\substack{U\subsetneq T\\ |U|=k-2\\}}\;  f_U(0^m) + f_U(a_i+b_i)=0.
\ee
Since $f_U$ obeys (I4) and $|U|$ is odd, we have
$f_U(X)=\sum_{V\subsetneq U} f_V(x)$ for $x=0^m$
or $x=a_i+b_i$. Thus Eq.~(\ref{extra_condition_7})
is equivalent to 
\be
\label{extra_condition_8}
\sum_{\substack{U\subsetneq T\\ |U|=k-2\\}}\; 
\sum_{V\subsetneq U} f_V(0^m) + f_V(a_i+b_i)=0.
\ee
All terms $f_V$ with odd $|V|$ vanish due
to (I5) since $|V|\le |U|-1=k-3$.
All terms $f_V$ with even $|V|$ and $|V|\le k-4$
vanish due to (I6).
Thus we can restrict the sum Eq.~(\ref{extra_condition_8}) to terms with 
$|V|=k-3$ and it suffices to check that 
\be
\label{extra_condition_9}
\sum_{\substack{U\subsetneq T\\ |U|=k-2\\}}\; 
\sum_{\substack{V\subsetneq U\\ |V|=k-3\\}}\; 
f_V(0^m) + f_V(a_i+b_i)=0.
\ee
However, each term $f_V(0^m)$ and $f_V(a_i+b_i)$
is counted exactly two times: if 
$V=T\setminus \{p,q\}$ then one can choose
$U=T\setminus \{p\}$ or $U=T\setminus \{q\}$.
Since we do all arithmetic modulo two,
this implies Eq.~(\ref{extra_condition_9}).
Thus (I4) is satisfied. 
This completes the induction step for odd $k$.

\end{document}